\documentclass[12pt]{amsart}
\usepackage{amssymb,amsfonts}
\usepackage{amsmath,amsthm,amsxtra}
\numberwithin{equation}{section}
\usepackage{bbm}
\usepackage[all]{xy}
\usepackage{graphicx}
\usepackage{color}
\usepackage{verbatim}
\usepackage[notcite,notref]{showkeys}

\usepackage{enumerate}

\usepackage[utf8]{inputenc}  
\usepackage[T1]{fontenc}

\newcommand\bigcheck[1]{#1 \raise1ex\hbox{$\hspace{-1ex}{}^\vee$}}
\newcommand\sucheck[1]{#1 \raise0.5ex\hbox{$\hspace{-1ex}{}^\vee$}}

\makeatletter

\renewcommand\section{\@startsection {section}{ 0}{\z@}%
                                    {-4ex \@plus -2ex \@minus -.2ex}%
                                    {6ex \@plus.2ex}%
                                    {\centering\large\scshape}}
\renewcommand\subsection{\@startsection{subsection}{1}{\z@}%
                                     {-5ex\@plus -6ex \@minus -3ex}%
                                     {-75ex \@plus 100ex}%
                                     {\normalfont\large\bfseries}}
\makeatother

\newtheorem{theorem}{Theorem}[section]
\newtheorem{lemma}[theorem]{Lemma}
\newtheorem{corollary}[theorem]{Corollary}
\newtheorem{proposition}[theorem]{Proposition}
\newtheorem*{lemma*}{Lemma}

\theoremstyle{definition}
\newtheorem{definition}[theorem]{Definition}

\theoremstyle{remark}

\newtheorem{remark}[theorem]{Remark}
\newtheorem{example}[theorem]{Example}

\newtheorem{conjecture}[theorem]{Conjecture}

\newcommand{\mc}[1]{{\mathcal #1}}

\renewcommand{\tilde}{\widetilde}

\usepackage[pdftex]{hyperref}
\hypersetup{
    bookmarks=true,         
    unicode=false,          
    pdftoolbar=true,        
    pdfmenubar=true,        
    pdffitwindow=true,      
    pdftitle={My title},    
    pdfauthor={Author},     
    pdfsubject={Subject},   
    pdfnewwindow=true,      
    pdfkeywords={keywords}, 
    colorlinks=false,       
    linkcolor=red,          
    citecolor=green,        
    filecolor=magenta,      
    urlcolor=cyan           
}
\begin{document}

\begin{titlepage}
	\centering
	\vspace{3.5cm}
	{\Huge\bfseries A sufficient condition for a Rational Differential Operator to generate an Integrable System\par}
	\vspace{2cm}
	{\Large\itshape Sylvain Carpentier*\par}
\vspace{1 cm}
\begin{abstract}
Hey !
\end{abstract}
	{\large \today\par}
\vspace{1 cm}
\textbf{Abstract}
\\
For a rational differential operator $L=AB^{-1}$, the Lenard-Magri scheme of integrability is a sequence of functions $F_n, n\geq 0$, such that (1) $B(F_{n+1})=A(F_n)$ for all $n \geq 0$ and (2) the functions $B(F_n)$ pairwise commute. We show that, assuming that property $(1)$ holds and that the set of differential orders of $B(F_n)$ is unbounded, property $(2)$ holds if and only if $L$ belongs to a class of rational operators that we call integrable. If we assume moreover that the rational operator $L$ is weakly non-local and preserves a certain splitting of the algebra of functions into even and odd parts, we show that one can always find such a sequence $(F_n)$ starting from any function in Ker B. This result gives some insight in the mechanism of recursion operators, which encode the hierarchies of the corresponding integrable equations.

{\let\thefootnote\relax\footnotetext{* \textit{Department of Mathematics, Massachusetts Institute of Technology, Cambridge, MA 02139, USA}}}
\end{titlepage}

\section{ Introduction}

In this paper, we work in the framework of an algebra of differential functions, introduced in [BDSK09], which is a  differential extension of the algebra of differential polynomials in one variable,  $R=\mathbb{C}[u,u',u",\dots]$ where the total derivative $\partial$ is defined by $\partial(u^{(n)})=u^{(n+1)}$. More specifically, an algebra $\mathcal{V}$ is called an \textit{algebra of differential functions} if (1) it contains $R$, (2) it is endowed with commuting derivations $\frac{\partial}{\partial u^{(n)}}$ extending the partial derivatives on $A$ in such a way that, for all function $f \in \mathcal{V}$, only finitely many of these extended partial derivatives of $f$ are nonzero and (3) there is a derivation $\frac{\partial}{\partial x}$ of $\mc V$ which is $0$ on $R$ and commute with all the $\frac{\partial}{\partial u^{(n)}}$. Then, 
\begin{equation}
\partial=\sum_{n\geq 0}{u^{(n+1)}\frac{\partial}{\partial u^{(n)}}}+\frac{\partial }{\partial x}
\end{equation}
is a derivation on $\mc V$ which we call again the \textit{total derivative}. It extends the total derivative on $R$. The set $\mathcal{C} \equiv \{v \in \mathcal{V}, \partial (v)=0\}$ is a subalgebra of $\mathcal{V}$ called the subalgebra of constants. We assume that $\mc V$ is a domain and denote its field of fractions by $\mc K$, which is endowed with a algebra of differential functions structure as well.
\\
\indent
 Let $\mathcal{V}$ be an algebra of differential functions. Given $F \in \mathcal{V}$, define the associated  \textit{evolutionnary vector field} $X_F$ by : 
\begin{equation}
X_F=\sum_{n\geq 0}{F^{(n)}\frac{\partial}{\partial u^{(n)}}} \hspace{1 mm}.
\end{equation}
$X_F$ is a derivation on $\mathcal{V}$ commuting with the total derivative $\partial$ and mapping $u$ to $F$.
 Note that $\partial=X_{u'}$. The commutator endows the space of evolutionnary vector fields on $\mathcal{V}$ with a Lie algebra structure. This holds thanks to the identity 
\begin{equation}
[X_F,X_G]=X_{X_F(G)-X_G(F)},
\end{equation}
satisfied for all $F,G \in \mathcal{V}$. We can identify $F \in \mathcal{V}$ with $X_F$ and in particular infer a Lie bracket $\{.,.\}$ on $\mathcal{V}$ defined by
\begin{equation}
\{F,G\}=X_F(G)-X_G(F) \hspace{1 mm} .
\end{equation}
\indent
 The \textit{differential order} $d(F)$ of a function $F$ is defined as the greatest integer $n$ for which $\frac{\partial F}{\partial u^{(n)}} \neq 0$.
\\
\indent
  A \textit{differential operator} on $\mathcal{V}$ is an element of the algebra $\mathcal{V}[\partial]$, where the multiplication is defined by $\partial  F=F \partial+F'$. If $F, G \in \mathcal{V}$, observe that 
\begin{equation}
X_G(F)=\sum_{n\geq 0}{G^{(n)}\frac{\partial F}{\partial u^{(n)}}}=D_F(G) \hspace{1 mm},
\end{equation}
where $D_F$ is the differential operator
$\sum_{n \geq 0} \frac{\partial F}{\partial u^{(n)}}{\partial}^n$, called  the \textit{Frechet derivative} of $F$. Two functions $F$ and $G$ are said to commute, or to be a symetry of one another, if 
\begin{equation}
\{F,G\}=0 \hspace{1 mm}.
\end{equation}
 \indent
 In this paper, an \textit{integrable system} is, by definition, an abelian Lie subalgebra of  the Lie algebra $\mathcal{V}$ with bracket $(0.4)$ which contains functions of arbitrarily high order. Finally, $F \in \mc V$ (or rather the evolution equation $\frac{du}{dt}=F$) is called \textit{integrable} if $F$ lies in an integrable system. 
\indent
It is known ([MS08], [IS80], [SS84]) that if $F$ is integrable, then there exists a so-called \textit{recursion operator} $L$, defined by the property that its Lie derivative along $F$ vanishes :
\begin{equation}
\mathcal{L}_F(L) : \hspace{1 mm}= X_F(L)-[D_F,L]=0 \hspace{1 mm} .
\end{equation}
 \indent
If $L$ is a differential operator satisfying equation $(0.7)$, it preserves the centralizer of $F$ in the Lie algebra $\mathcal{V}$. However, in almost all interesting cases, $L$ is not a differential operator but lies in an extension of $\mathcal{V}[\partial]$ called the algebra of \textit{pseudodifferential operators}, $\mathcal{V}((\partial^{-1}))$, where the multiplication is extended by letting for all $a \in \mathcal{V}$
\begin{equation}
\partial^{k}a=\sum_{n\geq 0} { \binom{k}{n} a^{(n)} \partial^{k-n}}, \hspace{1 mm} k \in \mathbb{Z} \hspace{1 mm}.
\end{equation}
\indent
 Most of the integrable systems encountered in the litterature are generated by such an operator $L$, namely they consist of the iterative images of a function $F_0$ by $L$ : $\{L^n(F_0)\}_{n \geq 0}$. Often, the existence of this operator is used as a non-rigorous argument to claim that the corresponding functions indeed define an abelian subalgebra of $\mathcal{V}$ . 
\\
\indent
Of course, for an arbitrary pseudodifferential operator $L$, the expression $L(F)$ makes no sense. However, for the rational pseudodifferential operators, i.e. those that can be written as the ratio of two differential operators $L=AB^{-1}$, we can define an image of $F$ by $L$ if $F$ lies in the image of $B$. Namely, we say, as in [DSK13], that $(G,F)$ are associated through $L$ if there exists a function $H$ such that 
\begin{equation}
(G,F)=(A(H),B(H)) \hspace{1 mm}.
\end{equation}
 In practice, we will just write $G=L(F)$, but we should always keep in mind that $G$ is only defined modulo $A(Ker B)$ and in particular is not unique. We denote the algebra of rational pseudodifferential operators by $\mathcal{V}(\partial)$ and simply call them \textit{rational operators}. We refer to [CDSK12] and [CDSK14] for an in-depth study of the algebra $\mathcal{V}(\partial)$.
\\
\indent
In this paper we are interested in determining which pairs of differential operators $(A,B)$ produce integrable systems. More precisely, for which $(A,B) \in \mc V[\partial]^2$ can we find a sequence of functions $(F_n)_{n \geq 0} \in \mc V^{\mathbb{Z}_+}$ such that :
\begin{equation}
\begin{split}
 A(F_n)=B(F_{n+1}) &\hspace{10 mm} \forall n \geq 0 \hspace{1 mm},\\
\{B(F_n),B(F_m)\}&=0 \hspace{3 mm} \forall n,m \geq 0 \hspace{1 mm}.
\end{split}
\end{equation}
\indent
It is often said in the litterature that a sufficient condition is for the rational operator $L=AB^{-1}$ to be \textit{Nihenjuis}, or \textit{hereditary}, which means that $L$ satisfies the following identity for all function $F \in \mc V$ :
\begin{equation} 
\mathcal{L}_{A(F)}(L)=L\mathcal{L}_{B(F)}(L) \hspace{1 mm}.
\end{equation}
Note that it follows from equations $(0.7)$ and $(0.11)$ that, if $L=AB^{-1}$ is hereditary and recursion for $B(F)$, then it is recursion for $A(F)$. 
Although necessary, $(0.11)$ is not sufficient for the pair $(A,B)$ to satisfy $(0.10)$ for some functions $F_n \in \mc V$. A counterexample is given by \begin{equation}
L=\partial^{-1}u''\partial \hspace{1 mm},
\end{equation}
which is hereditary but does not generate commuting functions. 
\\
\indent
We will see that, to satisfy $(0.10)$ for some functions $F_n \in \mc V$, the rational operator $L=AB^{-1}$ must lie in a finer subset of $\mc V(\partial)$ which we call the class of integrable rational operators. A differential operator $A$ is called \textit{integrable} if one can find a bidifferential skewsymmetric operator $M$ (i.e. an element of $\mc V[\partial_1,\partial_2]$) such that for all functions $F$ and $G$ one has 
\begin{equation}
X_{A(F)}(G)-X_{A(G)}(F)=A(M(F,G)) \hspace{1 mm}.
\end{equation}
\indent
We then define a pair of operators $(A,B)$ to be integrable if any of their linear combination is integrable, more precisely if there exists two skewsymmetric bidifferential operators $M$ and $N$ such that for all functions $F,G$ and all constant $\lambda$, we have
\begin{equation}
X_{(A+\lambda B)(F)}(G)-X_{(A+ \lambda B)(G)}(F)=(A+ \lambda B)((M+ \lambda N)(F,G)).
\end{equation}
\indent
For example, any local Poisson structure $H$ is an integrable operator, and any compatible pair of local Poisson structures $(H,K)$ is an integrable pair.
\\
\indent
Finally, a rational operator $L$ is called integrable if it can be written in the form $AB^{-1}$ where $(A,B)$ is an integrable pair of differential operators.
This definition is natural for the following reason :
\begin{theorem}
Let $A$, $B$ be two differential operators and $(H_n)_{n \geq 0}$ be a sequence of functions in $\mc V$ which spans an infinite dimensional space over $\mc C$. Let us assume moreover that $L$ is recursion for $B(H_0)$ and that for all $n \geq 0$,
\begin{equation}
B(H_{n+1})=A(H_n) \hspace{1 mm} .
\end{equation}
 Then the functions $B(H_n)$ pairwise commute if and only if the pair $(A,B)$ is integrable.
\end{theorem}
In the first section of the paper, we recall elementary properties of algebras of differential functions, differential operators, Frechet derivatives and bidifferential operators. In the second section, we begin by recalling and proving the result of Ibragimov and Shabat (stated in [IS80] and proved in [SS84]), which states that an integrable system has a recursion operator. Then, we move on studying hereditary operators and in particular their relation with the Lenard-Magri scheme of integrability $(0.10)$. In the third section, we introduce integrable operators and prove  one of their key properties :
\begin{proposition}
Let $L=AB^{-1}$ be an integrable rational operator which is recursion for $B(F)$. Then $A(F)$ and $B(F)$ commute :
\begin{equation}
\{A(F),B(F)\}=0.
\end{equation}
\end{proposition}
In the fourth section, we consider rational operators of a particular kind, namely, \textit{weakly non-local operators}, which by definition are rational operators which can be written in the following form (see [MN01]) 
\begin{equation}
L=E(\partial)+\sum_{i=1}^n{p_i \partial^{-1} q_i} \hspace{1 mm},
\end{equation}
where $E$ is a differential operator, called the local part of $L$, and $p_i$ and $q_i$ are elements of $\mc K$. It is shown in the paper that the space of weakly non-local operators coincides with the space of rational operators $AB^{-1}$ whose denominator $B$ in the minimal fractional decomposition (i.e. $deg B$ is minimal, see [CDSK12]) has full kernel in $\mc K$, i.e. :
\begin{equation}
deg B=dim_{\mathcal{C}}  \hspace{2 mm} Ker_{\mc K} B \hspace{1 mm}.
\end{equation}
\indent
Differential operators with property $(0.18)$ have been studied in [DSKT15] and called strongly non-degenerate there. The majority of known recursion operators associated to integrable systems are weakly non-local. Differential operators with full kernel provide an easy way to test whether a function $F$ lies in their image or not. Indeed, if $Ker B^*$ is spanned by $q_1,\dots , q_n$, then
\begin{equation}
F \in \hspace{2 mm} Im B \iff q_i F \in \partial \mc K \hspace{2 mm} \forall i=1\dots n \hspace{1 mm}.
\end{equation}
\indent
We study the structure of weakly non-local operators and examine what can be said when a rational operator $L \in \mc V(\partial)$ is integrable and weakly non-local. In particular, we prove :

\begin{proposition}
Let $L=E+ \sum_{i=1}^n{p_i \partial^{-1}q_i}$ be a weakly non-local operator where the $\{p_1,\dots,p_n \}$ and $\{q_1,\dots,q_n \}$ are sets consisting of linearly independent elements of $\mc V$. Then $L$ is integrable if and only if it is hereditary and $q_i$ is a variational derivative for all $i=1 \dots n$.
\end{proposition}
Recall that the \textit{variational derivative} of $F \in \mc V$ is defined by
\begin{equation}
\frac{\delta F}{\delta u} : \hspace{1 mm}= \sum_{n \geq 0} {(-\partial)^n (\frac{\partial F}{\partial u^{(n)}})} \hspace{1 mm} .
\end{equation}
\indent
In section $5$, we will assume that  the algebra of differential functions $\mathcal{V}$ can be split as a sum
\begin{equation}
\mathcal{V}=\mathcal{V}_{\bar{0}} \oplus \mathcal{V}_{\bar{1}}
\end{equation}
of eigenspaces of an involution $\sigma$ of the algebra $\mc V$, such that $\sigma \partial \sigma^{-1}=-\partial$ and $\sigma \frac{\partial}{\partial u^{(n)}} \sigma ^{-1}=(-1)^n \frac{\partial}{\partial u^{(n)}}$ for all $n \geq 0$. We call elements of $\mc V_{\bar{0}}$ even functions and elements of $\mc V_{\bar{1}}$ odd functions. It follows from this definition that $\partial$ switches parity and that $\frac{\delta}{\delta u}$ preserves parity. 
\\
\indent
We then give a sufficient condition for the existence of Lenard-Magri sequences for a subclass of rational operators :

\begin{theorem}
Let $L \in (\mc V [\partial])_{\bar{0}}+\mc V_{\bar{1}} \partial^{-1} \mc V_{\bar{0}}$ be integrable. If $AB^{-1}$ is a right minimal fractional decomposition of $L$ and $F_0 \in Ker B$, then there exists a sequence $F_n \in \mc V, n \geq 0$ such that 
\begin{enumerate}
\item[(1)] $B(F_{n+1})=A(F_n)$ for all $n\geq 0$ \hspace{1 mm},\\
\item[(2)] $\{B(F_n),B(F_m)\}=0$ for all $n,m \geq 0$\hspace{1 mm}.
\end{enumerate}
Moreover, if $d(L) >0$ and $d(B(F_N))$ is greater than the differential orders of the coefficients of $A$ and $B$ for some $N$, then the sequence $d(B(F_n))$ is unbounded.
\end{theorem}
A more complete picture is given by :
\begin{theorem}
Let $L \in (\mc V [\partial])_{\bar{0}}+\mc V_{\bar{1}} \partial^{-1} \mc V_{\bar{0}}$ be integrable. Then $L^k$ is weakly non-local and integrable for all $k \geq 1$. If 
\begin{equation}
L^k=E_k+\sum_{i=1}^{n_k}{p_{ki}\partial^{-1} q_{ki}} \hspace{1 mm},
\end{equation}
where $\{p_{k1},\dots,p_{kn_k}\}$ and $\{q_{k1},\dots,q_{kn_k}\}$ are two sets of linearly independent functions and $E_k$ is a differential operator,
then the functions $p_{ki}$ are odd, $E_k$ is even and $q_{ki}$ are even variational derivatives : $q_{ki}=\frac{\delta \rho_{ki}}{\delta u}$, with even $\rho_{ki}$. Moreover, for all $k,l \geq 0$ and $i,j \in \{1,n_k\} \times \{1,n_l\}$,
\begin{equation}
\begin{split}
 \{p_{ki},p_{lj}\}&=0, \hspace{1 mm} p_{ki} q_{lj}  \in \partial \mc V \hspace{1 mm}, \\
 \rho_{lj} \text{ is a conserved }& \text{density for the equation } u_t=p_{ki}.
\end{split}
\end{equation}
\end{theorem}
Theorems $0.4$ and $0.5$ also hold when the non-local part of $L$ lies in $\mc V_{\bar{0}} \partial^{-1} \mc V_{\bar{1}}$. Note also that by Proposition $0.3$ one could replace the assumption $L$ integrable by : $L$ is hereditary and the $q_i$'s are variational derivatives.
Theorem $0.5$ was proved in [EOR93] for the KdV equation.
\\
\indent
In section $6$, we explain how Theorem $0.4$ reproves the integrability of most of the integrable equations for which a recursion operator is known. In particular, the Korteweg-de Vries (KdV) equation
\begin{equation}
u_t=u'''+3uu'.
\end{equation}
satisfies the hypothesis of Theorems $0.4$ and $0.5$. Indeed, it admits the simple recursion operator
\begin{equation}
L_{KdV}=\partial^2+2u+u'\partial^{-1},
\end{equation}
which is known to be hereditary, as a ratio of two compatible local Poisson structures. Furthermore, the algebra $R$ of differential polynomials in $u$ admits a decomposition as in $(0.21)$ by declaring $u$ to be even and $\partial$ to be odd. For this splitting, $u'$ is odd, $\partial^2+u$ is an even operator and $1$ is an even variational derivative, hence $L_{KdV}$ is integrable by Propostition $0.3$. Therefore we can apply Theorem $0.4$ starting at $F_0=1$ and obtain an integrable hierarchy containing KdV, which we get at the second step of the Lenard-Magri scheme. Here the differential orders go to $+ \infty$ because $d(KdV)=3$ is greater than the differential orders of the coefficients of $L_{KdV}$.
\\
\indent
Similarly, the integrability of the Krichever-Novikov (KN) equation 
\begin{equation}
u_t=u'''-\frac{3}{2}\frac{u''^2}{u'}+\frac{P(u)}{u'}, \hspace{2 mm}  \text{where} \hspace{2 mm}  \frac{d^5 P}{du^5}=0 \hspace{1 mm} ,
\end{equation}
follows from Theorems $0.4$ and $0.5$. In [DS08], Demskoi and Sokolov exhibit a recursion operator for $(KN)$ of the form
\begin{equation}
L_{KN}=\partial^4+a_1\partial^3+a_2\partial^2+a_3\partial+a_4 +G_1 \partial^{-1} \frac{\delta \rho_1}{\delta u}+u'\partial^{-1} \frac{\delta \rho_2}{\delta u}.
\end{equation}
The space of Laurent differential polynomials in $u$, $\mc A=\mathbb{C}[u^{\pm 1},u'^{\pm 1},...]$ admits a decomposition as in $(0.21)$ by declaring $u$ to be even and $\partial$ to be odd. 
From the explicit formulas given in [DS08], it is straighforward to check that $a_i$ has the same parity as $i$ for $i=1,\dots,4$ and that $\rho_i$ are even for $i=1,2$. Hence, the local part of $L_{KN}$ is even and so are the functions $\frac{\delta \rho_i}{\delta u}$, since variational derivatives preserve parity. Moreover, $G_1$ is the equation (KN) itself, which is odd. Finally, one checks that $L_{KN}$ is hereditary and that the condition on the degrees ( last part of Theorem $0.4$) is met, hence $(0.26)$ is integrable, by Theorem $0.4$.
\\
\indent
 As we noted earlier, not all integrable systems admit a weakly non-local recursion operator. For instance, the Calogero-Degasperis-Ibragimov-Shabat (CDIS) equation,
\begin{equation}
\frac{du}{dt}=u'''+3u^2u''+9uu'^2+3u^4u' \hspace{1 mm} ,
\end{equation}
has a rational integrable recursion operator which is rational, but not weakly non-local :
\begin{equation}
L_{CDIS}=\frac{1}{u} \partial (\partial +2u^2)^{-1}(\partial+u^2-\frac{u'}{u})^2(\partial+2u^2){\partial}^{-1}u \hspace{1 mm} .
\end{equation}
\indent
For this particular equation, it is not hard to check that we can apply $L_{CDIS}$ infinitely many times to $u'$ and that $L_{CDIS}$ is integrable. Since $CDIS=L_{CDIS}(u')$, we conlude from Theorem $0.1$ that $CDIS$ is integrable. In the last section of the paper we show how to use our techniques to prove integrability of all equations from the classification list of [SW09].
\\
\indent
We conjecture that every integrable system comes with an integrable rational recursion operator :

\begin{conjecture}
Let $\mc W \subset \mc V$ be an integrable system. Then there exists an integrable rational operator $L=AB^{-1}$ such that 
\begin{equation}
\mc L_{F}(L)=0  \hspace{2 mm} \forall F \in \mc W \hspace{1 mm} .
\end{equation}
\end{conjecture}
Our results can be naturally extended to the case of several field variables. We plan on tackling these extensions in a future publication. 
\\
\indent
I would like to thank Victor Kac for suggesting the problem and for many helpful discussions. I am also very thankful to Alberto De Sole for his careful and patient reading of this paper.
\\
\\
\\
\\
\\
\\
\\
\\
\\
\vspace{18 cm}

\section{Preliminaries}
 We begin by recalling a few properties of algebras of differential functions, differential operators, Frechet derivatives and evolutionnary vector fields. For a more complete introduction to these notions, we refer the reader to [BDSK09] (for algebras of differential functions) and [MS08] (for Frechet derivatives and evolutionnary vector fields).
\subsection{Algebras of differential functions}
hey\\
\\
The basic differential algebra that we consider here is the algebra of differential polynomials in $u$, namely
\begin{equation}
R=\mathbb{C}[u,u',u'',\dots] \hspace{1 mm}.
\end{equation}
The total derivative $\partial$ on $R$ is defined by letting $\partial(u^{(n)})=u^{(n+1)}$ for all $n \geq 0$. An \textit{algebra of differential functions in u} is an extension of $R$ in the following sense :

\begin{definition}
An algebra of differential functions $\mc V$ in the variable $u$ is an algebra extension of $R$ endowed with derivations $(\frac{\partial}{\partial u^{(n)}})_{n \geq 0}$ and $\frac{\partial}{\partial x}$ such that
\begin{enumerate}
\item[(1)] $(\frac{\partial}{\partial u^{(n)}})_{ n \geq 0}$  extend the partial derivatives in $R$, \\
\item[(2)] The derivations $(\frac{\partial}{\partial u^{(n)}})_{ n \geq 0}$, $\frac{\partial}{\partial x}$ pairwise commute,\\
\item[(3)] For all $f \in \mc V$, $\frac{\partial f}{\partial u^{(n)}}=0$ \hspace{2 mm} for all but finitely many $n$, \\
\item[(4)] $\frac{\partial P}{\partial x}=0 \hspace{2 mm} \forall P \in R \hspace{1 mm}.$
\end{enumerate}
We define the total derivative $\partial$ on $\mc V$ by the following formula
\begin{equation}
\partial=\sum_{ n \geq 0} {u^{(n+1)}\frac{\partial}{\partial u^{(n)}}}+\frac{\partial}{\partial x} \hspace{1 mm},
\end{equation}
\end{definition}
which estends $\partial$ from $R$. For $f \in \mc V$ and $n \geq 0$, we will often write $f^{(n)}$ instead of $\partial^n(f)$, and $f, f',f'',...$ instead of $f, \partial(f), \partial^2(f),...$.
\\
\indent
Typical examples of algebras of differential functions that we will consider are: the algebra
of differential polynomials itself, any localization of it by some
element $f\in R$, or, more generally, by some multiplicative subset $S \subset R$, such as the whole field
of fractions $Q = \mathbb{C}(u^{(n)}| n \geq 0)$, or any algebraic extension of the algebra $R$ or of the
field $Q$ obtained by adding a solution of certain polynomial equation. An example of the latter
type is $\mc V = \mathbb{C}[\sqrt{u}^{\pm}, u', u'', . . . ]$  , obtained by starting
from the algebra $R$, adding the square root of the
element $u$, and localizing by $\sqrt{u}$.
\\
\indent
In all the sequel,  let $\mc V$ be an algebra of differential functions. 
\begin{definition}
Let $F \in \mc V$. We call differential order of $F$, which will be denoted by $d(F)$ or $d_F$, the following quantity :
\begin{equation}
d(F)=\max \{n \geq 0|  \frac{\partial F}{\partial u^{(n)}} \neq 0\} \hspace{1 mm}.
\end{equation}
If the set on the RHS of $(1.3)$ is empty, we call $F$ a \textit{quasiconstant} of the algebra $\mc V$. If $F$ is a quasiconstant such that $\frac{\partial F}{\partial x}=0$, we say that $F$ is a \textit{constant} of $\mc V$. We will denote the subset of quasiconstants (resp. constants) by $\mc Q_{\mc V}$ (resp. $\mc C_{\mc V}$), or when there is no confusion by $\mc Q$ (resp $\mc C$). We will always assume that $\mc C_{\mc V}$ is an algebraically closed field.
\end{definition}
 
\begin{remark}
For all $F\in \mc V$ such that $d(F) \geq 0$, we have $d(F')=d(F)+1$. This follows immediately from $(1.2)$ and the fact that partial derivatives do not increase the differential order. The latter fact holds because partial derivatives commute. In particular, $F' \neq 0$. Therefore 
\begin{equation}
Ker  \hspace{1 mm} \partial=\mc C_{\mc V} \hspace{1 mm}.
\end{equation}
\end{remark}

We denote the canonical projection $\mc V \rightarrow \mc V /{\partial \mc V}$ by $\int$. This notation is justified by the integration by parts property
\begin{equation}
\int{F'G}=-\int FG' \hspace{2 mm} \forall F,G \in \mc V \hspace{1 mm}.
\end{equation}
\indent
Given a function $F \in \mc V$ of differential order $N$ we call \textit{evolution equation associated to F} the following equation :
\begin{equation}
\frac{du}{dt}=F(u,u',\dots,u^{(N)}) \hspace{1 mm}.
\end{equation}
We also associate to $F$ a derivation on $\mc V$
\begin{equation}
X_F \equiv \sum_{ n \geq 0}{F^{(n)}\frac{\partial}{\partial u^{(n)}}} \hspace{1 mm},
\end{equation}
which is called the \textit{evolutionnary vector field} associated to $F$. Conversely, $F$ is said to be the characteristic function of $X_F$. 
\begin{lemma}
\begin{enumerate} 
\item[(1)] For any $F \in \mc V$, $[X_F,\partial]=0$.\\
\item[(2)] For all $F,G \in \mc V$, $[X_F,X_G]=X_{X_F(G)-X_G(F)}$.
\end{enumerate}
\end{lemma}
\begin{proof}
$(1)$ is a consequence of $(1.7)$ and the following identity
\begin{equation}
[\frac{\partial}{\partial u^{(n+1)}},\partial]=\frac{\partial}{\partial u^{(n)}} \hspace{2 mm} \forall n \geq 0 .
\end{equation}
which is immediate from $(1.2)$.
The terms involving products of partial derivatives vanishing in the commutator $[X_F,X_G]$, there are some functions $G_{n} \in \mc V$ such that
\begin{equation}
[X_F,X_G]=\sum_{n}{G_{n}\frac{\partial}{\partial u^{(n)}}}.
\end{equation}
\indent
From $(1)$ we know that $\partial$ commutes with $[X_F,X_G]$. Hence for all $n \geq 0$, $G_{n}={G_{0}}^{(n)}$ and $[X_F,X_G]=X_{G_0}$. Moreover, $G_{0}=[X_F,X_G](u)=X_F(G)-X_G(F)$, which completes the proof.
\end{proof}

From the second part of Lemma $1.4$, it follows that the following bracket is a Lie bracket on $\mc V$ :
\begin{equation}
\{F,G\} \equiv X_F(G)-X_G(F).
\end{equation}
Indeed $(1.10)$ is obviously bilinear over the constants and skewsymmetric. Furthermore, Jacobi identity holds because of $(2)$ in Lemma $1.4$.

\begin{definition}
$F, G \in \mc V$ are said to be symmetries of one another, or to commute, if 
\begin{equation}
\{F,G \}=0.
\end{equation}
\end{definition}

\begin{remark}
The use of the word symmetry is an abuse of language. What $(1.11)$ means is that, if $u$ is a solution of the equation $(1.6)$, then $u+ \epsilon G$ is a solution of $(1.6)$ as well modulo $\epsilon^2$ if and only if $(1.11)$ holds. A more rigorous denomination is generator of an infinitesimal symmetry. 
\end{remark}

\begin{definition}
An integrable system on $\mc V$ is an abelian subalgebra $\mathfrak g $ of $(\mc V, \{.,.\})$ such that the set of differential orders of elements of $\mathfrak g$ is unbounded.
\end{definition}

\begin{definition}
Given $F \in \mc V$, we say that $\rho \in \mc V$ is a conserved density of the evolution equation $\frac{du}{dt}=F$, or simply of $F$, if
\begin{equation}
\int X_F(\rho)=0.
\end{equation}
When $\rho \in \partial \mc V$, $(1.12)$ holds because the derivation $X_F$ commutes with $\partial$. In that case, we say that $\rho$ is a trivial conserved density.
By integration by parts, $(1.12)$ is equivalent to
\begin{equation}
\int \frac{\delta \rho}{\delta u}.F=0,
\end{equation}
where the function $\frac{\delta \rho}{\delta u}$ is called the variational derivative of $\rho$. It is given by
\begin{equation}
 \frac{\delta \rho}{\delta u} \equiv \sum_n (-\partial)^n(\frac{\partial \rho}{\partial u^{(n)}}).
\end{equation}
\end{definition}
\begin{definition}
For $n \geq 0$ let $\mc V_n$ be the subset of functions in $\mc V$ whose differential order is at most $n$.
We say that $\mc V$ is normal if $\frac{\partial}{\partial u^{(n)}}(\mc V_{n})=\mc V_{n}$ for all $n \geq 0$ and $\frac{\partial}{\partial x}(\mc Q_{\mc V})=\mc Q_{\mc V}$.
\end{definition}

\begin{lemma}
If $\mc V$ is normal, then
\begin{equation} 
Ker \frac{\delta}{\delta u}=\partial \mc V .
\end{equation}
\end{lemma}
\begin{proof}
See Proposition $1.5$ in [BDSK]. In our case $\mc Q_{\mc V} \subset \partial \mc V$.
\end{proof}

\subsection{Differential operators}
h 
\\
\\
In this subsection, we will be reviewing definitions and results from $[CDSK12]$ and $[CDSK13]$. We refer the reader to these papers for the proofs. Let $\mc K$ be a differential field with subfield of constants $\mc C$.

\begin{definition}
A differential operator $A(\partial)$ on $\mc K$ is an element of the algebra $\mc K[\partial]$ in which multiplication is defined by $\partial a=a \partial +a'$ for every $a \in \mc K$. A pseudodifferential operator $A(\partial)$ on $\mc K$ is an element of the algebra 
$\mathcal{K}((\partial^{-1}))$ in which multiplication is defined by $\partial^{-1}a=\sum_{n\geq 0} { (-1)^n a^{(n)} \partial^{-n-1}}$ for all $a \in \mathcal{V}$. The algebra $\mc K ((\partial^{-1}))$ is a skewfield. A rational ( pseudodifferential ) operator on $\mathcal{K}$ $A(\partial)$ is a pseudodifferential operator which can be written as the ratio of two differential operators $B(\partial)$ and $C(\partial)$ : $A=BC^{-1}$. We denote the set of rational operators on $\mc K$ by $\mc K (\partial)$. The following inclusions are obvious
\begin{equation}
\mc K[\partial] \subset \mc K(\partial) \subset \mc K ((\partial^{-1}))
\end{equation}
If $L=\sum_{-\infty}^N{l_n \partial^{n}} \in \mc K((\partial^{-1}))$ with $l_N \neq 0$, we call $N=d(L)$ the degree of $L$.
\end{definition}

\begin{lemma}
$\mc K(\partial)$ is a field, which can be described also in terms of left fractions :
\begin{equation}
\mc K ( \partial)=\{AB^{-1}|(A,B) \in \mc K[\partial] \times \mc K[\partial]^{\times}\}=\{D^{-1}C|(C,D) \in \mc K[\partial] \times \mc K[\partial]^{\times}\}
\end{equation}
\end{lemma}
\begin{proof}
Part (a) of Proposition $3.4$ in [CDSK12].
\end{proof}

\begin{definition}
For $L \in \mc K(\partial)$ we call right (resp. left) fractional decomposition a pair $(A,B) \in \mc K[\partial]\times K[\partial]^*$ such that $L=AB^{-1}$ (resp. $L=B^{-1}A$).
\end{definition}

\begin{lemma}
Let $L \in \mc K(\partial)$. Then :
\begin{enumerate}
\item[(1)] There exists a right fractional decompostion $(A_0,B_0)$ of $L$ such that
for any other right fractional decomposition $(A,B)$ of $L$ there exists a non-zero differential operator $D$ such that $A=A_0D$ and $B=B_0D$. We call $(A_0,B_0)$ a minimal right fractional decomposition of $L$. \\
\item[(2)]
The analogous statement for left fractions holds.
\\
\item[(3)] Finally,
if $D_0^{-1}C_0$ is a minimal left fractional decomposition of $L$, then $d(D_0)=d(B_0)$.
\end{enumerate}
\end{lemma}
\begin{proof}
$(1)$ and $(2)$ follows from part (b) of Proposition $3.4$ in [CDSK12]. $(3)$ can be found in Remark $3.8$ of [CDSK13].
\end{proof}

\begin{lemma}
Let $A$ be a non-trivial differential operator on $\mc K$. Then 
\begin{equation}
dim_{\mc C} ker A \leq d(A).
\end{equation}
Moreover, $Im A$ is infinite-dimensional over $\mc C$.
\end{lemma}
\begin{proof}
The set of differential orders of $(A(u^{(n)}))_{n \geq 0}$ is clearly unbounded, hence $Im A$ is a infinite-dimensional space over $\mc C$. As for the statement $(1.18)$, see  $A.3.5$ (a) in [DSK13].
\end{proof}

\begin{definition}
The ring $\mc K[\partial]$ is right and left principal ideal domain ([CDSK13b]). If $A$ and $B$ are differential operators, we call left (resp. right) least common multiple of the pair $(A,B)$ a generator of the left (resp. right) ideal of $\mc K[\partial]$ generated by $A$ and $B$. 
\end{definition}

\begin{lemma}
Let $f_1,\dots,f_n \in \mc V$ be linearly independent over $\mc C$. Then there exists a differential operator $P \in \mc V[\partial]$ with degree $n$ such that Ker P$=\langle f_1,\dots,f_n \rangle$.
\end{lemma}
\begin{proof}
Let us prove the claim by induction on $n$. If $n=1$, we define $P=f_1\partial-f_1'$. It is clear that $P$ has degree $1$ and that its kernel is spanned over $\mc C$ by $f_1$. Let us consider $n+1$ linearly independent functions $f_1,\dots,f_{n+1}$ and let $Q$ be a degree $n$ differential operator whose kernel is spanned by $f_1,\dots,f_n$. By construction of $Q$, $Q(f_{n+1}) \neq 0$, hence we can define $R=(Q(f_{n+1})\partial-Q(f_{n+1})')Q$. We can see that $f_i \in  Ker R$ for $i=1,\dots,n+1$. We conclude that $Ker R$ is spanned by the $f_i$'s using equation $(1.18)$ and noting that $deg(R)=n+1$.
\end{proof}

\subsection{Frechet Derivatives}
hey\\
\\
From now on, we consider an algebra of differential functions $\mc V$ over $u$. Furthermore, we assume that $\mc V$ is a domain and let $\mc K$ be its field of fraction. Recall that $\mc C$ is assumed to be algebraically closed. For a complete discussion on Frechet derivatives, we send the interested reader to [MS08].
\begin{definition}
Let $F \in {\mc V}$. We define the Frechet derivative of $F$ to be the differential operator $D_F$ such that 
\begin{equation}
D_F(G)=X_G(F) \hspace{4 mm} \forall G \in {\mc V}.
\end{equation}
Consequentely, for $F \in \mc V$ we have
\begin{equation}
D_{F}=\sum_{m}{\frac{\partial F}{\partial u^{(m)}}\partial^m}.
\end{equation}
\end{definition}

\begin{definition}
Let $P= \sum_{k=0}^N{p_k \partial^k}$ be a differential operator over $\mc V$ and $F \in \mc V$. We define the differential operator $(D_P)_F$ as follows
\begin{equation}
(D_{P})_F=\sum_{k=1}^{N}{ {F^{(k)} D_{p_k}(\partial)}}.
\end{equation}
If $L= \sum_{-\infty}^N{l_k \partial^k} \in \mc V ((\partial^{-1}))$ and $F \in \mc V$, we let
\begin{equation}
X_F(L)=\sum_{-\infty}^N{X_F(l_k) \partial^k}.
\end{equation}
\end{definition}

\begin{lemma}
The following identities hold for all differential operators $A,B \in \mc V[\partial]$ and functions $F,G \in \mc V$ :
\begin{equation}
D_{A(F)}=AD_F+(D_A)_F.
\end{equation}
\begin{equation}
(D_{AB})_F=(D_A)_{B(F)}+A(D_B)_F.
\end{equation}
\begin{equation}(D_A)_F(G)=X_G(A) (F).
\end{equation}
\end{lemma}

\begin{proof}
$(1.23)$ follows from two facts. First, $D_{ab}=aD_b+bD_a$ for all functions $a,b \in \mc V$ since partial derivatives are derivations of $\mc V$. Second,we have $D_{a'}=\partial D_a$ for all function $a \in \mc V$, which can be inferred from the identity $(1.8)$.
\\
\indent
Let us consider the three following specializations of $(1.23)$ :
\begin{equation}
\begin{split}
D_{AB(F)}&=ABD_F+(D_{AB})_F,\\
D_{A(B(F))}&=AD_{B(F)} + (D_A)_{B(F)}, \\
D_{B(F)}&=BD_F+(D_B)_F.
\end{split}
\end{equation}
We find $(1.24)$ by multiplying the third line of $(1.26)$ on the left by $A$, adding the second line and substracting the first. 
As for $(1.25)$ we first note that, since $X_G$ is a derivation of $\mc V$ commuting with $\partial$, we have for all $G \in \mc V$
\begin{equation}
X_G(A(F))=X_G(A)(F)+A(X_G(F)).
\end{equation}
Then, using the definition $(1.19)$ of the Frechet derivative, we have
\begin{equation}
D_{A(F)}(G)=X_G(A)(F)+AD_F(G).
\end{equation}
We obtain $(1.25)$ comparing $(1.28)$ with $(1.23)$.
\end{proof}

\begin{definition}
Let $*$ be the anti-involution of $\mc V((\partial^{-1}))$ such that $\partial^{*}=-\partial$ and $F^*=F$ for all $F \in \mc V$. It is well defined since the relation $(0.9)$ is preserved. 
\end{definition}

\begin{lemma}
Let  $\mc V$ be normal and $F \in \mc V$. Then 
$D_F=D_F^*$ if and only if $F$ is a variational derivative :
\begin{equation}
D_F=D_F^* \iff \exists \rho \in \mc V, F=\frac{\delta \rho}{\delta u}.
\end{equation}
\end{lemma}
\begin{proof}
See Proposition $1.9$ in [BDSK09].
\end{proof}

\begin{lemma}
Let $F,G \in \mc V$. Then
\begin{equation}
X_F(D_G)=[D_F,D_G]+X_G(D_F)+D_{\{F,G\}}.
\end{equation}
\end{lemma}
\begin{proof}
Since $\mc V$ is a Lie algebra for $\{.,.\}$, Jacobi identity holds for $(F,G,H)$ for all $H \in \mc V$. Therefore
\begin{equation}
\begin{split}
0&=\{F,\{G,H\}\}+\{G,\{H,F\}\}+\{H,\{F,G\}\} \\
 &= (X_F-D_F)(\{G,H\})-(X_G-D_G)(\{F,H\})-(X_{\{F,G\}}-D_{\{F,G\}})(H) \\
 &= (X_F-D_F)(X_G-D_G)(H)-(X_G-D_G)(X_F-D_F)(H)\\
   & \hspace{6 mm}-(X_{\{F,G\}}-D_{\{F,G\}})(H)\\
&= [D_F,D_G](H)+X_G(D_F)(H)-X_F(D_G)(H) \\
 & \hspace{6 mm}+D_{\{F,G\}}(H)+([X_F,X_G]-X_{\{F,G\}})(H) \\
&=([D_F,D_G]+X_G(D_F)-X_F(D_G)+D_{\{F,G\}})(H).
\end{split}
\end{equation}
Since this is true for all $H$ we obtain the desired result.
\end{proof}

\subsection{Bidifferential operators}

\begin{definition} A bidifferential operator on $\mc V$ is an element $M$ of $\mc V [\partial_1,\partial_2]$ :
\begin{equation}
M(\partial_1,\partial_2)=\sum_{k,l} {M_{kl}{\partial_1^k}\partial_2^l}.
\end{equation}
where $M_{kl}$ are functions in $\mc V$. It naturally defines a map from $\mc V \otimes_{\mc C} \mc V$ to $\mc V$ : for any $F,G \in \mc V$, we let
\begin{equation}
M(F,G)=\sum_{k,l}{M_{kl}F^{(k)}G^{(l)}}.
\end{equation}
 For a bidifferential operator $M$ and a function $F \in \mc V$ we define differential operators $M_F$ and $M^F$ by
\begin{equation}
M_F=\sum_{k,l} {M_{kl}F^{(k)}\partial^l}.
\end{equation}
\begin{equation}
M^F=\sum_{k,l} {M_{kl}F^{(l)}\partial^k}.
\end{equation}
By construction, we have for all $F,G \in \mc V$
\begin{equation}
M_F(G)=M^G(F)=M(F,G).
\end{equation}
For a bidifferential operator $M$ we let
\begin{equation}
\begin{split}
  d_1(M) &= \underset{F \in \mc V}{\sup} \hspace{1 mm} {d(M_F)}, \\
  d_2(M) &= \underset{F \in \mc V}{\sup} \hspace{1 mm}{ d(M^F)}.
\end{split}
\end{equation}
For a bidifferential operator $M$ and a differential operator $B$ we define the bidifferential operators $BM$ and $MB$ by letting for all $F,G \in \mc V$ :
\begin{equation}
\begin{split}
(BM)(F,G)&=B(M(F,G)),\\
(MB)(F,G)&=M(F,B(G)).
\end{split}
\end{equation}
\end{definition}

\begin{example}
We constructed in  Definition $1.19$ a bidifferential operator $D_A$ given a differential operator $A$. We call $D_A$ the Frechet derivative of $A$.
\end{example}

\begin{lemma}
Let $M$ be a bidifferential operator and $B$ be a differential operator on $\mc V$. Then there exists a unique pair $(N,P)$ of bidifferential operators on $\mc K$ such that
\begin{enumerate}
\item[(1)] $M=BP+N$, \\
\item[(2)] $d_1(N) < d(B)$.
\end{enumerate}
Similarly, there exists a unique pair $(Q,R)$ of bidifferential operators on $\mc K$ such that
\begin{enumerate}
\item[(1)] $M=QB+R$, \\
\item[(2)] $d_2(R) < d(B)$.
\end{enumerate}
\end{lemma}
\begin{proof}
The uniqueness is obvious. Let us show the existence for the left division. If $d_1(M)<d(B)$ there is nothing to do. Otherwise, let $b_k \partial^{k}$ be the leading term of $B$ and $A(F)\partial^l$ be the leading term of $M_F$.
Then if we let $\tilde{M}_F=M_F-B\frac{A(F)}{b_k} \partial^{l-k}$, we have $d_1(\tilde{M})<d_1(M)$. We conclude by induction on $d_1(M)$.
\end{proof}

\begin{lemma}
Let $M$ and $N$ be two bidifferential operators over $\mc V$, $A$ and $B$ two  differential operators such that for all $F \in \mc V$
\begin{equation}
M_{A(F)}=BN_F.
\end{equation}
Then there exists a bidifferential operator $P$ on $\mc K$ such that for all $F \in \mc V$
\begin{equation}
N_F=P_{A(F)}.
\end{equation}
Moreover $M=BP$.
\end{lemma}
\begin{proof}
Let $k=d(B)$. We prove the existence of $P$ by induction on $d_1(N)$. If $d_1(N)=0$, then the leading term of $(1.39)$ reads $M_k(A(F))=b_k N_0(F)$ hence $N_0$ is divisble on the right by $A$, i.e. $(1.40)$ is satisfied for some bidifferentila operator $P$. If $d_1(N)=l>0$, the leading term of $(1.39)$ is $M_{k+l}(A(F))=b_k N_l(F)$ hence there exists a differential operator $Q$ such that $N_l=Q(A(F))$. Therefore we have
\begin{equation}
(M-B(Q\partial^l))_{A(F)}=B(N-N_l \partial^l)_F,
\end{equation}
and we can use the induction hypothesis to conclude that $(1.40)$ holds for some bidifferential operator $P$. 
Combining $(1.39)$ with $(1.40)$ gives
\begin{equation}
M_{A(F)}=BP_{A(F)}
\end{equation}
Finally, since the image of $A$ is infinite-dimensional by Lemma $1.15$, and since the kernel of the map $F \mapsto Q_F$ is finite dimensional for any bidifferential operator $Q$, we can remove $A$ from $(1.42)$ and deduce that $M=BP$.
\end{proof}

\begin{lemma}
Let $M$ and $N$ be two bidifferential operators and $A$ and $B$ be two differential operators on $\mc V$ such that $AM_F=N_F B$ for all $F \in \mc V$. Then there exists a bidifferential operator $P$ on $\mc K$ such that $M=PB$ and $N=AP$.
\end{lemma}
\begin{proof}
Performing Euclidean divisions of $M$ by $B$ on the right we can reduce the problem to the case where $d_1(M)<d(B)$. Let us show that in that case $M=0$. Let $\mc L$ be a differential field extension of $\mc K$ with  $\mc C_{\mc K}=\mc C_{\mc L}$, where $B$ has a full kernel. Such an extension exists by the Picard-Vessiot theory. In $\mc L$, we have equality in $(1.18)$. Let $b \in Ker_{\mc L} B$. We know that for all $F \in \mc V$, $A(M_F(b))=A(M^b(F))=0$. In the language of differential operators, this means that $AM^b=0$, therefore that $M^b$=0. Thus, for all $F \in \mc V$ the kernel of $M_F$ in $\mc L$ has dimension over $\mc C$ at least $d(B)$. This implies that $M_F=0$, because we are in the case where $d(B)>d_1(M)$. Hence $M=0$.
\end{proof}

\begin{lemma}
Let $A$ and $B$ be two differential operators and $AD=BC$ be their least right common multiple. Let $M$ and $N$ be bidifferential operators such that $AM=BN$. Then there exists a bidifferential operator $P$ such that $M=DP$ and $N=CP$.
\end{lemma}
\begin{proof}
This lemma holds when $M$ and $N$ are differential operators by the very definition of the least right common multiple. Hence, for all $F \in \mc V$, we know that $M_F$ is divisible on the left by $D$. By Lemma $1.26$, this means that $M$ is divisible on the left by $D$.
\end{proof}

\section{Hereditary Operators}
Let $\mc V$ be a normal algebra  of differential functions in $u$. Let us assume that $\mc V$ is a domain and let $\mc K$ be its field of fraction. Let $\mc C$ be the subfield of constants, assumed to be algebraically closed.
\subsection{Recursion operators}
\begin{definition}
The Lie derivative $\mc L_F (L)$ of a pseudodifferential operator $L$ along the function $F$ is defined to be 
\begin{equation}
\mc L_F (L)=X_F(L)-[D_F,L] .
\end{equation}
 We say ([Olv93]) that $L$ is recursion for $F$ if 
\begin{equation}
\mc L_F(L)=0.
\end{equation}
\end{definition}

\begin{remark}
Every pseudodifferential operator $L$ is recursion for $0$ and $u'$. The constant operator $1$ is recursion for any function $F$. The Lie derivative $\mathcal{L}_F$ is a derivation of the algebra of pseudodifferential operators $\mathcal{V}((\partial^{-1}))$, because both $X_F$ and $[D_F,.]$ are.
\end{remark}

\begin{proposition}
Let $F \in \mc V$, and $H=AB^{-1}$, $K=CD^{-1}$ be two (non-local) Poisson structures on $\mc V$ (see [DSK13] for the definition) such that $F=A(X)=C(Y)$, where $X,Y \in \mc V$ and $B(X)$ and $D(Y)$ are variational derivatives. Then, $\Lambda=HK^{-1}$ is a recursion operator for $F$.
\end{proposition}
\begin{proof}
Since $B(X)$ is a variational derivative, we have by $(1.29)$: $D_{B(X)}=D_{B(X)}^*$. We specialize formula $(ii)$ of Proposition $6.8$ in [DSK13] for the pairs $(A,B)$ and $(X,G)$, where $G \in \mc V$, to get
\begin{equation}
A^*(D_{B(G)}(F)+D_F^*(B(G)))+B^*(D_{A(G)}(F)-D_F(A(G)))=0.
\end{equation}
Using $(1.23)$, $(1.25)$, and the fact that $A^*B+B^*A=0$ (Poisson structures are skewadjoint), we have for all $G \in \mc V$
\begin{equation}
A^*(X_F(B)(G)+D_F^*(B(G)))+B^*(X_F(A)(G)-D_F(A(G)))=0.
\end{equation}
Therefore we obtain the following differential operator identity :
\begin{equation}
A^*X_F(B)+A^*D_F^*B+B^*X_F(A)-B^*D_FA=0.
\end{equation}
Multiplying $(2.5)$ on the left by $(B^*)^{-1}$, on the right by $B^{-1}$, and using the skewadjointness of $H$, we deduce that
\begin{equation}
X_F(H)=D_F H+HD_F^* .
\end{equation}
Since the same argument holds for $K$, we conclude that
\begin{equation}
X_F(HK^{-1})=[D_F,HK^{-1}],
\end{equation}
which means that $HK^{-1}$ is a recursion operator for $F$.
\end{proof}

Next, we address the problem of existence of a recursion operator for an integrable system (cf. [MS08], [IS80], [SS84]).

\begin{proposition}
Let $\mathcal{W}$ be an integrable system in $\mathcal{V}$ spanned by countably many $(F_n)_{n \geq 1}$ such that their differential orders $d_{F_n}$ are strictly increasing and $d_{F_1} \geq 2$. Let us assume moreover that the leading term of $D_{F_1}$ is invertible in $\mc V$ and admits a $d_{F_1}$-th root in $\mc V$. Then there exists a pseudodifferential operator $L$ in $\mc V(( \partial ^{-1}))$ which is recursion for all the $F_n$.
\end{proposition}

We will first state and prove some lemmas before proving Proposition $2.4$.

\begin{lemma}
Let $A$ be a pseudodifferential operator of degree $d_A$ and $F \in \mc V$. Then for all  non-zero integer $i \in \mathbb{Z}$, 
\begin{equation}
d(\mc L_F (A^{i}))= d(\mc L_F(A))+(i-1)d(A).
\end{equation}
\end{lemma}
\begin{proof}
Let us first consider a positive integer $i$.
As we noted in Remark $2.2$, $\mc L_F$ is a derivation of $\mathcal{V}((\partial^{-1}))$, hence 
\begin{equation}
\mc L_F (A^{i})= \sum_{k=0}^{i-1}{A^k \mc L_F(A) A^{i-1-k}}.
\end{equation}
In the RHS of $(2.9)$, all terms have degree $d(\mc L_F(A))+(i-1)d_A$ and have the same leading coefficient, therefore the result follows. The same holds if we replace
$A$ by its inverse $A^{-1}$. It remains to check that $d(\mc L_F (A^{-1}))=d(\mc L_F (A))-2d_A$. This follows from the identity
\begin{equation}
 \mc L_F(A^{-1})=-A^{-1} \mc L_F(A) A^{-1}.
\end{equation}
\end{proof}

\begin{lemma}
Let $F \in \mc V$ and $ G \in \mc V$ be two functions with $G$ not constant and $d_F \geq 2$. Then
\begin{equation}
d(\mc L_F(G)) = d_F-1.
\end{equation}
\end{lemma}
\begin{proof}
$X_F(G)$  differential operator of degree $0$ or $-\infty$. Moreover, since $G$ is not a constant, the degree of the commutator $[D_F,G]$ is $d_F-1$.
\end{proof}

\begin{lemma}
Let $F$ be a function of differential order $d_F \geq 2$ and $(T,S)$ be two pseudodifferential operators of degree $1$ such that $d(\mathcal{L}_F(S))=m<d(\mathcal{L}_F(T))=n$. Then, there exists a Laurent series $\phi (z) \in \mathbb{C}((z^{-1}))$ of degree $1$ such that $T-\phi(S)$ has degree $n-d_F+1$. Furthermore, $d(\mathcal{L}_F(\phi(S)))=d(\mathcal{L}_F(S))$.
\end{lemma}
\begin{proof}
Let $T=\sum_{-\infty <i \leq 1} {f_i}S^i$ be the expansion of $T$ in terms of powers of $S$, where $f_i \in \mc V$. Such an expansion exists because $S$ has degree $1$. Using Leibniz rule, we get  
\begin{equation}
\mathcal{L}_F(T)=\sum{\mathcal{L}_F ({f_i})S^i }+ \sum{{f_i}\mathcal{L}_F(S^i)}.
\end{equation}
 Note that $f_1$ is non-zero because both $S$ and $T$ have degree $1$. Moreover, according to Lemma $2.5$ for all integer $i$ $d(\mathcal{L}_F (S^i))= m+i-1$. Therefore $d(\sum{{f_i}\mathcal{L}_F (S^i)})=m$. As a direct consequence of $(2.12)$, we obtain 
\begin{equation}
d(\mc L_F(T))=d(\sum{\mathcal{L}_F ({f_i})S^i })=n.
\end{equation}
 Let $k=\max\{l,f_i '\neq 0\}$. Since the Lie derivative vanishes on constant operators there exists a nonconstant $f_i$, or we would contradict $(2.13)$. Using Lemma $2.6$ we deduce 
\begin{equation}
d(\sum{\mathcal{L}_F ({f_i})S^i })=k+d_F-1.
\end{equation}
 and, from $(2.13)$, we see that $n=k+d_F-1$. In other words, 
\begin{equation}
T=\sum_{i=n-d_F+2}^{1}{c_i S^i}+ \sum_{i\leq n-d_F+1}{f_i S^i},
\end{equation}
 where $c_i$ are constants and $f_{n-d_F+1}$ is not. If $n=d_F$, we can take any Laurent series $\phi$ of degree $1$ and the lemma will hold. Otherwise, define the Laurent polynomial 
\begin{equation}
\phi (z)=c_1 z+ \dots c_{n-d_F+2}z^{n-d_F+2}.
\end{equation}
Comparing $(2.15)$ with $(2.16)$ and remembering that $f_{n-d_F+1} \neq 0$, we deduce that $T-\phi(S)$ has degree $n-d_F+1$. Finally, by Lemma $2.5$,
\begin{equation}
 d(\mathcal{L}_F(\phi(S)))=d(\mathcal{L}_F(S)).
\end{equation}
\end{proof}

\begin{lemma}
Let $F$ be a function in $\mc V$ such that $d_F \geq 2$ and $T_n \in \mc V((\partial^{-1}))$ be a sequence of  pseudodifferential operators of degree one with invertible leading terms such that the sequence $d(\mathcal{L}_F(T_n))$ is strictly decreasing. Then there exists a sequence of complex laurent series $\phi_n$ of degree $1$ such that 
the sequence $\phi_n(T_n)$ admits a limit $L$ in $\mc V((\partial^{-1}))$. 
\end{lemma}

\begin{proof}
Let us denote  $d(\mathcal{L}_F(T_n))$ by $d_n$ and rename $T_{n}^{1}$ the sequence of operators. Making use of  Lemma $2.6$ we can find Laurent series of degree $1$ ${\psi}_{n}^{1}(z)$ for $n\geq 2$ such that all ${\psi}_{n}^{1} (T_{n}^{1})-T_{1}^{1}$ have degree $d_1-d_F+1$. We denote $\psi_n^1(T_n^1)$ by $T_{n}^{2}$ for $n \geq 2$ and let $T_{1}^{2}=T_{1}^{1}$. Because we assumed that the leading term of the $T_n$ is invertible in $\mc V$, the sequence $T_n^2$ lies in $\mc V((\partial^{-1}))$. Moreover, since $\psi_n^1$ are Laurent series of degree $1$ in $\mathbb{C}$, for all $n \geq 2$ we have by Lemma $2.5$ :
\begin{equation}
d(\mc L_F(T_n^1))=d(\mc L_F(T_n^2)).
\end{equation}
 Next, we use Lemma $2.6$ one more time to obtain Laurent series of degree $1$ ${\psi}_{n}^{2}(z)$ for $n\geq 3$ such that ${\psi}_{n}^{2} (T_{n}^{2})-T_{2}^{2}$ has degree $d_2-d_F+1$ for all $n \geq 3$. As previously,  let us denote by $T_{n}^{3}$ the operator ${\psi}_{n}^{2} (T_{n}^{2})$ for $n \geq 3$, and let $T_{1,2}^3=T_{1,2}^2$. Since $d_1> d_2$, for all $n\geq 2$, $T_n^3-T_1^1=(T_n^3-T_2^2)+(T_2^2-T_1^1)$ has degree $d_1-d_F+1$. Moreover, for all $n \geq 3$ we have
\begin{equation}
d(\mc L_F(T_n^3))=d_n.
\end{equation}
 Iterating the argument, we construct a sequences of pseudodifferential operators $T_{n}^{k}$ with coefficients in $\mc V$ and Laurent series $(\psi_n^k)_{n \geq k+1}$ such that for all $m,k>n$,
\begin{equation}
d(T_n^n-T_m^k) =d_n-d_F+1,
\end{equation}
and such that for all $k \geq n$
 \begin{equation}
d(\mc L_F(T_n^k))=d_n.
\end{equation}
Specializing $(2.20)$ to $k=m$ we get for $n<m$
\begin{equation}
d(T_n^n-T_m^m) =d_n-d_F+1.
\end{equation}
 Therefore $T_n^n=\psi_n^1 \dots \psi_n^{n-1}(T_n^1)=\phi_n(T_n)$ admits a limit $L$ in $\mc V((\partial^{-1}))$. 
\end{proof}

\begin{lemma}
Let $F$ and $G$ be two commuting functions, with $d_G \geq 1$. Then, if $\mu_F$ denotes the degree of $[D_F,\partial]$, and provided that $D_G$ admits a $d_G$-th root in $\mc V(( \partial ^{-1}))$,
\begin{equation}
d(\mathcal{L}_F(D_G^{1/d_G}))=\mu_F-d_G+1.
\end{equation}
\end{lemma}
\begin{proof}
By Lemma $1.23$, and since $F$ and $G$ commute, we have
\begin{equation}
  X_F(D_G)=[D_F,D_G]+X_G(D_F).
\end{equation}
In other words, 
\begin{equation}
\mathcal{L}_F(D_G)=X_G(D_F).
\end{equation}
Therefore $d(\mathcal{L}_F(D_G))=d(X_G(D_F))$. Note that, since $d_G \geq 1$, $X_G(H) \neq 0$ for any nonconstant function $H$. Hence the degree of $X_G(D_F)$ is the highest integer $n$ for which $\frac{\partial F}{\partial u^{(n)}}$ is not a constant. This is the same as the degree of $[D_F,\partial]$. Therefore
\begin{equation}
d(\mc L_F(D_G))=\mu_F.
\end{equation}
 We complete the proof combining Lemma $2.5$ with $(2.26)$.
\end{proof}

\begin{lemma}
Let $\mathcal{W}$ be an integrable system in $\mathcal{V}$ spanned by countably many functions $(F_n)_{n \geq 1}$ such that their differential orders $d_{F_n}$ are strictly increasing and $d_{F_1} \geq 2$. Let us assume moreover that $\frac{\partial F_1}{\partial u^{(d_{F_1})}}$ admits a $d_{F_1}$-th root in $\mc V$. Then $D_{F_n}$ admits a $d_{F_n}$-th root in $\mc V((\partial^{-1}))$ for all $n \geq 1$.
\end{lemma}
\begin{proof}
A pseudodifferential operator $L=l_N \partial^N
+ \dots \in \mc V((\partial^{-1}))$ admits a $N$-th root if and only if $l_N$ does in $\mc V$. Hence we want to show that $\frac{\partial F_n}{\partial u^{(d_{F_n})}}$ admits a $d_{F_n}$-th root for all $n$. We know it is the case when $n=1$. Let us denote by $X$ a $d_{F_1}$-th root of $\frac{\partial F_1}{\partial u^{(d_{F_1})}}$ in $\mc V$.  By degree considerations in $(2.24)$, we see that $[D_{F_n},D_{F_1}]$ has degree at most $d_{F_n}+d_{F_1}-2$. Hence the coefficient of order $d_{F_n}+d_{F_1}-1$ in the commutator $[D_{F_n},D_{F_1}]$ must vanish, i.e.
\begin{equation}
d_{F_n}\frac{\partial F_n}{\partial u^{(d_{F_n})}}(\frac{\partial F_1}{\partial u^{(d_{F_1})}})'=d_{F_1}\frac{\partial F_1}{\partial u^{(d_{F_1})}}(\frac{\partial F_n}{\partial u^{(d_{F_n})}})',
\end{equation}
which is to say that there is a contant $\alpha$ such that
\begin{equation}
(\frac{\partial F_n}{\partial u^{(d_{F_n})}})^{d_{F_1}}=\alpha (\frac{\partial F_1}{\partial u^{(d_{F_1})}})^{d_{F_n}}.
\end{equation}
and therefore
\begin{equation}
\frac{\partial F_n}{\partial u^{(d_{F_n})}}=(\alpha^{1/d_{F_1}}X)^{d_{F_n}}.
\end{equation}
\end{proof}

\begin{proof}, of Proposition $2.4.$
\\
 Let $T_n=D_{F_n}^{1/d_{F_n}}$, which are well-defined in $\mc V((\partial^{-1}))$ thanks to Lemma $2.10$. Moreover, they have the same leading term which is invertible in $\mc V$ because the leading term of $D_{F_1}$ is. By Lemma $2.8$, we have for $n \geq 1$
\begin{equation}
d(\mc L_{F_1}(T_n))=\mu_{F_1}-d_{F_n}+1.
\end{equation}
These degrees are strictly decreasing, hence we can apply Lemma $2.8$ to find Laurent series of degree $1$ $\phi_n$ and a pseudodifferential operator $L \in \mc V((\partial^{-1}))$ such that $\phi_n(T_n)$ converges to $L$.
 Taking the limit $m \rightarrow \infty$ in $(2.22)$, we get
\begin{equation}
d(\phi_n(T_n)-L)=\mu_{F_1}-d_{F_1}+2-d_{F_n}.
\end{equation}
 Let us check that $\mathcal{L}_{F_k}(L)=0$ for all $k$. First of all, using Lemma $2.9$, we have 
\begin{equation} 
d(\mathcal{L}_{F_k}(T_n))=\mu_{F_k}-d_{F_n}+1.
\end{equation} 
 Since $\mathcal{L}_{F_k}$ is a derivation and $\phi_n$ is a complex Laurent series of degree $1$, the same holds if we replace $T_n$ by $\phi_n(T_n)$
\begin{equation} 
d(\mathcal{L}_{F_k}(\phi_n(T_n)))=\mu_{F_k}-d_{F_n}+1.
\end{equation} 
 Moreover, by definition of the Lie derivative, it is clear that for any pseudodifferential operator $M$ and function $H$,
 \begin{equation}
d(\mathcal{L}_H(M))\leq d(M)+d_H.
\end{equation}
Finally, combining equations $(2.31)$ to $(2.34)$ : 
\begin{equation} 
\begin{split}
d(\mc L_{F_k}(L)) & \leq d(\mc L_{F_k}(\phi_n(T_n)))+d(\mc L_{F_k}(L-\phi_n(T_n))) \\ 
                          & \leq d(\mc L_{F_k}(\phi_n(T_n)))+d(L-\phi_n(T_n))+d_{F_k} \\
                          & \leq \mu_{F_k}-d_{F_n}+1+\mu_{F_1}-d_{F_1}+2-d_{F_n}+d_{F_k} \\
                          & \leq (\mu_{F_1}+\mu_{F_k}-d_{F_1}+d_{F_k}+3)-2d_{F_n}.
\end{split}
\end{equation}
The sequence of degrees $d_{F_n}$ is strictly increasing. Therefore for all $k \geq 1$ $L$ is recursion for $F_k$ :
\begin{equation}
\mathcal{L}_{F_k}(L)=0.
\end{equation}
\end{proof}
\begin{remark} Let $\mathcal{W}$ be an integrable system in $\mc V$. Then the quotient space $\mathcal{W}/(\mathcal{W} \cap \mathcal{V}_{1})$ admits a countable basis. Indeed, it follows from Lemma $2.10$ that if two functions $F$ and $G$ of differential order $n \geq 2$ commute, then $\frac{\partial F}{\partial u^{(n)}}$ and $\frac{\partial G}{\partial u^{(n)}}$ are proportional.
\end{remark}

\subsection{Hereditary operators}

\begin{definition}
A rational operator $L=AB^{-1} \in \mc V (\partial)$ is called \textit{hereditary} (cf. [O93]) if the following identity holds for all $F \in \mc V$
\begin{equation}
\mathcal{L}_{A(F)}(L)=L\mathcal{L}_{B(F)}(L).
\end{equation}
\end{definition}
\begin{remark}
Using, $(1.23)$, we can rewrite the preceding definition as follows. An operator $L=AB^{-1}$ is hereditary if and only if for all $F \in \mc V$,
\begin{equation}
X_{A(F)}(L)-[(D_A)_F, L]=L (X_{B(F)}(L)-[(D_B)_F, L]) .
\end{equation}
Indeed, $[AD_F,AB^{-1}]=AB^{-1}[BD_F,AB^{-1}]$ for all $F$.
\end{remark}

\begin{remark}
In the definition of hereditariness, the choice of the fractional decomposition does not matter. If $A=A_0 X$, $B=B_0X$ and equation $(2.38)$ holds for $(A,B)$, it also holds for $(A_0,B_0)$. Indeed, every coefficient of $(2.38)$ as a pseudodifferential operator on $\mc V$ is a differential operator on $F$. Since they vanish on the image of $X$, which is infinite dimensional over $\mc C$, they are identically zero.
\end{remark}
The following lemma gives a reason for the name hereditary : if $L$ is recursion for $G \in \mc V$, it is also recursion for $L(G)$.
\begin{lemma}
 If $L=AB^{-1}$ is a rational operator recursion for $B(F)$ and hereditary, it follows from Definition  $2.12$ that $L$ is recursion for $A(F)$ as well.
\end{lemma}
\begin{proof}
Obvious.
\end{proof}

\begin{example}
Here are a few hereditary operators :
\begin{enumerate}
\item[(a)] any matrix rational operator with constant coefficients.  In that case, both sides of $(2.38)$ vanish.\\
\item[(b)] $L=\partial(\partial+u)\partial^{-1}=\partial+u+u'\partial^{-1}$. In that case, $A=\partial(\partial+u)$ and $B=\partial$. Moreover, we have 
\begin{equation}
\begin{split}
 (D_A)_F &=\partial F\\
 (D_B)_F &=0 \\
                                X_{A(F)}(L) &= \partial A(F) \partial^{-1} \\
                                X_{B(F)}(L) &= \partial F' \partial ^{-1}.
\end{split}
\end{equation}
Therefore, for $L$, $\partial^{-1}(2.38)\partial$ is equivalent to 
\begin{equation}
A(F)-F \partial(\partial+u)+(\partial+u)F \partial=(\partial+u)F'.
\end{equation}
The coefficient of $\partial^2$ is $0$ on both sides of $(2.40)$. The constant coefficients also agree, as
\begin{equation}
 (F'+Fu)'-Fu'=F''+uF'.
\end{equation}
It remains to check that the coefficients in front of $\partial$ are the same in both sides of $(2.40)$. This follows from 
\begin{equation}
-Fu+F'+uF=F'.
\end{equation}
\item[(c)] $L=\partial^2+2u+u'\partial^{-1}$. It can be checked by a direct computation as in $(b)$.
\end{enumerate}
\end{example}

\begin{lemma}
Let $L=AB^{-1}$ be  a hereditary rational operator. Then $L^N$ is hereditary for any $N \in \mathbb{Z}$.
\end{lemma}
\begin{proof}
First, let us show that $L^{-1}$ is itself hereditary. Note that $L^{-1}=BA^{-1}$. Since $\mc L_F$ is a derivation we have for $F \in \mc V$
\begin{equation}
\begin{split}
   X_{B(F)}(L^{-1})&=-L^{-1}X_{B(F)}(L)L^{-1} \\
                          &=-L^{-2}(LX_{B(F)}(L))L^{-1} \\
                          &=L^{-1}(-L^{-1} X_{A(F)}(L) L^{-1}) \\
                          &=L^{-1}X_{A(F)}(L^{-1}).
\end{split}
\end{equation}
  It remains to consider the case $N \geq 2$. Let $A_1=A$, $B_1=B$, $C_0=0$ and $D_0=0$. Assuming to have defined $A_k$, $B_k$, $C_{k-1}$ and $D_{k-1}$ for $k \geq 1$ we let $C_k$ and $D_k$ such that $BC_k=A_kD_k$ and let $A_{k+1}=AC_k$ and $B_{k+1}=B_kD_k$. It is easy to show by induction that $L^k=A_k{B_k}^{-1}$. Let us show by induction on $k$ that 
\begin{equation}\mc L_{A_k(F)} (L)=L^k \mc L_{B_k(F)} (L).
\end{equation}
 By definition of hereditariness this is true for $k=1$. As for the induction step, 
\begin{equation} 
\begin{split}
\mc L_{A_{k+1}(F)} (L) & =  \mc L_{AC_k(F)} (L) \\ 
                                  & = L\mc L_{BC_k(F)} (L) \\
                                  & =  L \mc L_{A_kD_k(F)} (L) \\
                                  & =  L^{k+1} \mc L_{B_kD_k(F)} (L) \\
                                  & = L^{k+1} \mc L_{B_{k+1}(F)} (L).
\end{split}
\end{equation}
To conclude note that $\mc L_{A_k(F)}(L^k)=\sum_{i=0}^{k-1}{L^i \mc L_{A_k(F)}(L) L^{k-1-i}}$. Hence
\begin{equation}
\begin{split}
\mc L_{A_k(F)} (L^k)&=L^k\sum_{i=0}^{k-1}{L^i \mc L_{B_k(F)}(L) L^{k-1-i}}\\
                               &L^k \mc L_{B_k(F)} (L^k).
\end{split}
\end{equation}
\end{proof}

The following proposition says that a rational operator must be heredtary in order to generate an integrable system.

\begin{proposition}
Let $L=AB^{-1}$ be a rational operator and $(H_n)_{n \geq 0}$ a sequence of functions, which span an infinite-dimensional space over $\mc C$, such that 
\begin{enumerate}
 \item[(1)]  $B(H_{n+1})=A(H_n)$ for all $n \geq 0$, \\
  \item[(2)]  $\{B(H_n),B(H_m)\}=0$ for all $n,m \geq 0$.
\end{enumerate}
Then $L$ is recursion for all the $B(H_n)$ and hereditary.
\end{proposition}
\begin{proof}
Let $F=B_{H_n}$ for some $n \geq 0$. We know that $F$ commutes with $A(H_m)$ for all $m \geq 0$. This means by definition that
\begin{equation}
X_{F} (A(H_m))=D_{F} (A(H_m)) \hspace{2 mm} \forall m \geq 0.
\end{equation}
If we denote $-X_F(H_m)$ by $x_m$, we can rewrite  $(2.47)$ as 
\begin{equation}
(X_{F}(A)-D_{F}A)(H_m)=A(x_m) \hspace{2mm} \forall m \geq 0.
\end{equation}
 Similarly, because $F$ commutes with $B(H_m)$ for all $m \geq 0$, we see that 
\begin{equation}
(X_{F}(B)-D_{F}B)(H_m)=B(x_m) \hspace{2 mm} \forall m \geq 0.
\end{equation} 
 Let $C$ and $D$ be differential operators such that $CA=DB$. Then, multiplying on the left $(2.48)$ by $C$ and $(2.49)$ by $D$, we get that 
\begin{equation}
C(X_{F}(A)-D_{F}A)(H_m)=D(X_{F}(B)-D_{F}B)(H_m) \hspace{2 mm} \forall m \geq 0.
\end{equation}
 Since the sequence $H_m$ spans an infinite-dimensional space over $\mc C$, $(2.50)$ implies that
\begin{equation}
C(X_{F}(A)-D_{F}A)=D(X_{F}(B)-D_{F}B),
\end{equation}
which rewrites into
\begin{equation}
(X_{F}(A)-D_{F}A)=AB^{-1}(X_{F}(B)-D_{F}B).
\end{equation}
because $L=AB^{-1}=C^{-1}D$. From there, we can conclude that $L$ is recursion for $F$. Indeed
\begin{equation}
\begin{split}
      X_F(AB^{-1})&=X_F(A)B^{-1}-AB^{-1}X_F(B)B^{-1}\\
                          &=(D_FA+AB^{-1}(X_F(B)-D_FB))B^{-1}-AB^{-1}X_F(B)B^{-1} \\
                          &=D_FAB^{-1}-AB^{-1}D_F\\
                         &=[D_F,AB^{-1}],
\end{split}
\end{equation}
where we used $(2.52)$ to obtain the second line of $(2.53)$.
\\
\indent
Since $L$ is recursion for $A(H_n)$ and $B(H_n)$ for all $n \geq 0$, equation $(2.38)$ holds for $F=H_n$ for all $n \geq 0$. Moreover, since the $H_n$'s span an infinite-dimensional space over $\mc C$, we get that $(2.38)$ holds for all $H \in \mc V$ and that $L$ is hereditary.
\end{proof}

\begin{proposition}
Let $A( \partial ) \in \mc V [\partial]$ be a hereditary differential operator, and $F \in \mc V$ be such that $A$ is recursion for $F$. Then 
\begin{equation}
\{A^n(F),A^m(F)\}=0 \hspace{3 mm} \forall n,m \geq 0.
\end{equation}
\end{proposition}
\begin{proof}
 Since $A$ is recursion for $F$, it follows from iterating Lemma $2.15$ that $A$ is recursion for $A^n(F)$ for all $n \geq 0$. Let $n \geq 0$. We want to prove by induction on $m \geq 0$ that $A^n(F)$ and $A^{n+m}(F)$ commute. It is obviously true when $m=0$. Let us assume that, for some $m \geq 0$, $A^n(F)$ and $A^{n+m}(F)$ commute. Since $A$ is recursion for $A^n(F)$, we have 
\begin{equation}
\begin{split}
        \{A^n(F),A^{m+1}(F)\}&=X_{A^n(F)}(A^{m+1}(F))-X_{A^{m+1}(F)}(A^n(F))\\
                                              &=(X_{A^n(F)}(A)-[D_{A^n(F)},A])(A^{m}(F))\\
                                              & \hspace{2 mm}+A(\{A^n(F),A^{m}(F)\})\\
                                              &=0+0.
\end{split}
\end{equation}
\end{proof}
\begin{proposition}
Let $A=\sum_{k=0}^N{a_k \partial^k}$ be a hereditary differential operator on $\mc V$. Then for all $0 \leq k \leq N$, 
\begin{equation}
d(a_k) \leq N+1.
\end{equation}
\end{proposition}
\begin{proof}
$(2.38)$ for $B=1$ means that for all $F \in \mc V$,
\begin{equation}
X_{A(F)}(A)-[(D_A)_F ,A]=AX_F(A).
\end{equation}
Both $X_{A(F)}(A)$ and $AX_F(A)$ have degree in $\partial$ at most $2N$ hence 
\begin{equation}
d_{\partial}([(D_A)_F ,A]) \leq 2N \hspace{4 mm} \forall F \in \mc V.
\end{equation}
Let $m= \max_{k}{d(a_k)}$. Then
the coefficient of $\partial^{N+m-1}$ in $[(D_A)_F,A]$ is
\begin{equation}
m(\sum_k{F^{(k)}\frac{\partial a_k}{\partial u^{(m)}}})a_N'-Na_N(\sum_k{F^{(k)}\frac{\partial a_k}{\partial u^{(m)}}})',
\end{equation}
which is non-zero for $F \in \mc V$ with sufficiently large differential order. Therefore
\begin{equation}
m \leq N+1.
\end{equation}
\end{proof}

\section{Integrable Operators}
Hereditariness is not a sufficient condition for a rational operator to generate an integrable system. Let us consider the operator 
\begin{equation}
L=\partial^{-1}u'' \partial \hspace{1 mm}.
\end{equation}
We begin by finding the minimal right fractional decomposition of  $L$. Since $\partial ( \frac{u''}{u^{'''}} \partial-1)=u''\partial \frac{1}{u^{'''}}\partial $ we have
\begin{equation}
L=(\frac{u''}{u^{'''}} \partial-1)(\frac{1}{u^{'''}}\partial)^{-1}=AB^{-1}.
\end{equation}
Therefore we can compute $D_A$ and $D_B$. Namely, for all $F$
\begin{equation}
\begin{split}
 (D_A)_F&=\frac{F'}{u'''}(-\frac{u''}{u'''}\partial+1)\partial^2,\\
 (D_B)_F&=-\frac{F'}{{u'''}^2}\partial^3.
\end{split}
\end{equation}
We want to check that $L$ is hereditary, i.e. that $(2.38)$ holds for all $F \in \mc V$ :
\begin{equation}
X_{A(F)}(L)-[(D_A)_F,L]=L(X_{B(F)}(L)-[(D_B)_F,L]).
\end{equation}
Let us plug in $(3.1)$ and $(3.3)$ into $(3.4)$, multiply on the left by $\partial$ and on the right by $\partial^{-1}$. $(3.4)$ is equivalent to the following
\begin{equation}
A(F)''-[\partial \frac{F'}{u'''}(-\frac{u''}{u'''}\partial+1)\partial,u'']=u''B(F)''-u''[-\partial \frac{F'}{{u'''}^2} \partial^2,u''].
\end{equation}
Putting $u''$ inside the commutator and rearranging the terms yields
\begin{equation}
\begin{split}
A(F)''-u''B(F)''&=[(\partial \frac{F'}{u'''}(-\frac{u''}{u'''}\partial+1)+u''\partial\frac{F'}{{u'''}^2} \partial)\partial,u'']\\
                   &= [(\partial \frac{F'}{u'''}-\frac{F'}{u'''} \partial)\partial,u''] \\
                   &= (\frac{F'}{u'''})'u''' \\
                   &=u'''B(F)'.
\end{split}
\end{equation}
Hence to check that $L$ is hereditary amounts to check that the following identity holds for all $F$ :
\begin{equation}
A(F)'=u''B(F)',
\end{equation}
which is satisfied. Therefore $L$ is hereditary.
\\
\indent
However, $L$ is recursion for a function $F$ if and only if $F$ lies in the $\mc C$-span of $1$ and $u'$. Indeed, recall that
$L$ is recursion for a function $F$ if  and only if
\begin{equation}
X_F(\partial^{-1} u'' \partial)=[D_F,\partial^{-1} u'' \partial].
\end{equation}
Taking the conjugation of $(3.8)$ by $\partial$ gives
\begin{equation}
F''=[\partial D_F \partial^{-1},u''].
\end{equation}
Since the degree of $[M,G]$ is $d(M)-1$ when $M$ is a pseudodifferential operator and $G$ is a non-constant function, the degree of $D_F$ has to be at most $1$. From there it is easy to conclude that the only functions for which $L$ is a recursion operator are linear combination of $u'$ and $1$.
\\
\indent
This example shows that to be hereditary for a rational operator does not imply that it can generate infinitely many functions. A rational operator which generates an integrable system falls into a finner subclass of operators which we call integrables and study in detail in this section. 
\\
\indent
We now proceed to define integrable operators. We will show that any positive power of integrable $L$ is integrable, and that integrability is a necessary condition to generate infinitely many commuting functions. We will also state and prove one of the key properties of these operators which is that if integrable $L=AB^{-1}$ is recursion for $B(F)$, then $A(F)$ and $B(F)$ commute.
\begin{definition}
A differential operator $A \in \mc K[\partial]$ is called integrable if there exists a bidifferential operator $M$ on $\mc K$ such that for all $F \in \mc V$
\begin{equation}
X_{A(F)}(A)=D_{A(F)}A+A(M_F-D_F A).
\end{equation}
Or equivalently,
\begin{equation}
X_{A(F)}(A)-(D_A)_FA=AM_F \hspace{1 mm}.
\end{equation}
\end{definition}

\begin{example}
A Poisson differential operator $H$ is integrable. Indeed recall from [BDSK09] $(1.50)$ that $H$ is a Poisson differential operator if and only if it is skewadjoint $H^*=-H$ and satisfies  for all $(F,G)$
\begin{equation}
H(D_GH(F)+D_{H(F)}^*(G)-D_FH(G)+D_F^*H(G))=\{H(F),H(G)\}.
\end{equation}
Moreover, since $H^*=-H$, we have
\begin{equation}
D_{H(F)}^*(G)+D_F^*H(G)=(D_H)_F ^* (G),
\end{equation}
 hence for all $(F,G)$, the following equation holds
\begin{equation}
X_{H(F)}(H)(G)+HD_F(G)-D_{H(F)}H(G)=H(D_H)_F ^* (G).
\end{equation}
Rearranging the LHS of $(3.14)$ using $(1.24)$, we get for all $F$ :
\begin{equation}
X_{H(F)}(H)-(D_H)_F H=H (D_H)_F ^*.
\end{equation}
\end{example}

Not all integrable operators are Poisson :
\begin{example}
The operator $A=\partial(\partial+u)$ is integrable. Indeed $(3.11)$ holds for $A$ and all function $F \in \mc V$
\begin{equation}
\partial (F''+(uF)')-\partial F\partial(\partial+u)=\partial (\partial+u)(-F\partial+F').
\end{equation}
\end{example}

\begin{remark}
The bidifferential operator $M$ in $(3.11)$ has to be skewsymmetric. Indeed, thanks to $(1.25)$, $(3.10)$ and $(3.11)$ are equivalent to
\begin{equation}
X_{A(F)}(A)(G)-X_{A(G)}(A)(F)=A(M(F,G)) \hspace{3 mm} \forall F,G \in \mc V.
\end{equation}
\end{remark}

\begin{definition}
We say that a pair of differential operators $(A,B)$ on $\mc K$ is integrable if there exists two bidifferential operators $M$ and $N$ such that for all $F \in \mc V$ and for all $\lambda \in \mc C$, 
\begin{equation}
X_{(A+\lambda B)(F)}(A+ \lambda B)-(D_{A+ \lambda B})_F(A+\lambda B)=(A+ \lambda B)(M+ \lambda N)_F.
\end{equation}
We say that a rational operator $L \in \mc K(\partial)$ is integrable if there exists an integrable pair $(A,B)$ such that $L=AB^{-1}$.
Note that $(3.18)$ is stronger than asking for $(A+\lambda B)$ to be integrable for all $\lambda$. $(3.18)$ is equivalent to the following three equations, symmetric in $A$ and $B$, for all $F \in \mc V$ :
\begin{equation}
\begin{split}
AM_F&=X_{A(F)}(A)-(D_A)_FA\\
BN_F&=X_{B(F)}(B)-(D_B)_FB \\
  AN_F+BM_F&=X_{A(F)}(B)+X_{B(F)}(A)-(D_A)_F B-(D_B)_F A.
\end{split}
\end{equation}
\end{definition}

\begin{remark}
A compatible pair of local Poisson structures $(H,K)$ is integrable. Indeed, $(3.15)$ holds with $H+ \lambda K$ for any constant $\lambda$. Comparing with $(3.18)$, we obtain the claim.
\end{remark}

\begin{lemma}
Let $A, B, C$ be differential operators such that $A=BC$ is integrable. Then $B$ is integrable. 
\end{lemma}
\begin{proof}
We know that there exists a bidifferential operator $M$ such that for all $F \in \mc V$
\begin{equation}
 X_{A(F)}(A)=D_{A(F)}A+A(M_F-D_F A).
\end{equation}
 Since $A=BC$ and evolutionnary vector fields act via $(1.22)$ as derivations of the ring of differential operators, we obtain for all $F \in \mc V$
\begin{equation}
X_{BC(F)}(B)C+BX_{BC(F)}(C)-D_{BC(F)}BC+BCD_FBC=BCM_F.
\end{equation}
Let us start by simplifying the two last terms of the LHS using $(1.23)$ twice :
\begin{equation}
\begin{split}
D_{BC(F)}&=BD_{C(F)}+(D_B)_{C(F)}\\
              &=BCD_F+B(D_C)_F+(D_B)_{C(F)}.
\end{split}
\end{equation}
Combining $(3.21)$ and $(3.22)$, we obtain
\begin{equation}
B(CM_F+(D_C)_F BC-X_{BC(F)}(C))=(X_{BC(F)}(B)- (D_B)_{C(F)}B)C.
\end{equation}
We are in the situation where $BQ_F=P_{C(F)}C$ for some bidifferential operator $Q$ and 
\begin{equation}
P_F=X_{B(F)}(B)- (D_B)_{F}B.
\end{equation}
Lemma $1.27$ implies that there exists a bidifferential operator $R$ such that $P=BR$. This precisely means that $(3.11)$ holds, i.e. that $B$ is an integrable operator.
\end{proof}

Next, we make a connection with the previous notion of hereditary operator by showing that integrability implies hereditariness.

\begin{lemma}
Let $L$ be an integrable rational operator. Then $L$ is hereditary.  Conversely if $L$ is hereditary, $AB^{-1}$ is a minimal right fractional decompoisition of $L$ and $B$ is integrable, then the pair $(A,B)$ is integrable.
\end{lemma}
\begin{proof}
Let $(A,B)$ be a pair of operators satisfying $(3.18)$ and let $M, N$ be the bidifferential operators corresponding to $A$ and $B$. Let us condider the LHS of equation $(2.38)$ :
\begin{equation}
\begin{split}
 LHS &=X_{A(F)}(AB^{-1})-[(D_A)_F, AB^{-1}] \\
&=(X_{A(F)}(A)-(D_A)_F A )B^{-1}-AB^{-1}(X_{A(F)}(B)- (D_A)_F B) B^{-1} \\
 &=AB^{-1}(BM_F-X_{A(F)}(B)+(D_A)_F B)B^{-1} \\
 &=AB^{-1}(X_{B(F)}(A)-(D_B)_F A-AN_F)B^{-1}\\
 &=AB^{-1}(X_{B(F)}(A)-(D_B)_F A -AB^{-1}(X_{B(F)}(B)-(D_B)_F B))B^{-1}\\
 &=AB^{-1}(X_{B(F)}(AB^{-1})-[(D_B)_F, AB^{-1}])\\
 &= RHS.
\end{split}
\end{equation}
We used the first line of $(3.19)$ to go from line $2$ to line $3$, the third line of $(3.19)$ to deduce line $4$ from line $3$, and finally the second line of $(3.19)$ to go from line $4$ to line $5$.
\\
\indent
If we now simply assume that $L$ is hereditary, that $AB^{-1}$ is minimal and that $B$ integrable with corresponding matrix bidifferential operator $N$, we can say from the four last lines of $(3.25)$ that
\begin{equation}
RHS=AB^{-1}(X_{B(F)}(A)-(D_B)_F A-AN_F)B^{-1},
\end{equation}
where RHS denotes the right hand side of $(2.38)$. Similarly we can rewrite the LHS of $(2.38)$ as we did in $(3.25)$
\begin{equation}
LHS=(X_{A(F)}(A)-(D_A)_F A )B^{-1}-AB^{-1}(X_{A(F)}(B)- (D_A)_F B) B^{-1} \hspace{1 mm}.
\end{equation}
By hereditariness of $L$, the quantities in $(3.26)$ and $(3.27)$ are equal. In other words, we have the equation
\begin{equation}
P=AB^{-1}Q,
\end{equation}
where for all $F \in \mc V$,
\begin{equation}
\begin{split}
 P_F&=X_{A(F)}(A)-(D_A)_F A\\
 Q_F&=X_{A(F)}(B)+X_{B(F)}(A)-(D_A)_F B-(D_B)_F A-A N_F \hspace{1 mm}.
\end{split}
\end{equation}
Let $CA=DB$ be the left least common multiple of the pair $(A,B)$. It is also the right least common multiple of the pair $(C,D)$. Moreover $(3.28)$ rewrites as $CP=DQ$. Applying Lemma $1.29$ provides us with a bidifferential operator $M$ such that
\begin{equation}
\begin{split}
P&=AM\\
Q&=BM \hspace{1 mm}.
\end{split}
\end{equation}
Therefore the pair $(A,B)$ satisfies $(3.19)$ and is integrable.
\end{proof}

\begin{lemma}
Let $A$ and $B$ be integrable differential operators such that $AB^{-1}$ is hereditary. Then the pair $(A,B)$ is integrable. 
\end{lemma}
\begin{proof}
If we look carefully at the second part of the proof of Lemma $3.8$, we notice that we did not use the minimality of the pair $(A,B)$ to reach $(3.28)$. From there, if we know that $A$ is integrable, we can conclude as well the integrability of $(A,B)$. Indeed, since $A$ is integrable, $P$ is divisible on the left by $A$ by $(3.11)$ and we can cancel by $A$ in $(3.28)$.
\end{proof}

\begin{corollary}
Let $A$, $B$ and $C$ be differential operators. If the pair $(AC,BC)$ is integrable, then so is $(A,B)$. In particular, for an integrable operator $L$ with right minimal fractional decomposition $A_0 B_0^{-1}$, we have that $(A_0,B_0)$ is integrable.
\end{corollary}
\begin{proof}
By Lemma $3.8$ $AB^{-1}=(AC)(BC)^{-1}$ is hereditary and by Lemma $3.7$, both $A$ and $B$ are integrable since $AC$ and $BC$ are. We conclude using Lemma $3.9$. If $L$ is integrable there exist an integrable pair of differential operators $(A,B)$ such that $L=AB^{-1}$. If $A_0 B_0^{-1}$ is a minimal right fractional decomposition of $L$, we can find an operator $E$ such that $A=A_0 E$ and $B=B_0 E$. Therefore the pair $(A_0,B_0)$ is integrable as well.
\end{proof}

\begin{example}
The operator $L=\partial^{-1}u'' \partial$ is not integrable. Indeed, from the previous remark and by $(3.2)$ we should have that $\frac{1}{u'''}\partial$ is integrable. However, this is not the case, which will follow from Proposition $4.5$. Therefore hereditariness does not imply integrability.
\end{example}

The two following propositions aim at showing that the set of integrable operators is stable under taking integer powers. This property will turn out to be essential in proving the commutativity of the functions generated by the rational operator.

\begin{proposition}
Let $A$ and $B$ be two differential operators such that $(A,B)$ is integrable, and 
let $AD=BC$ be their right least common multiple (see Definition $1.16$). Then $(AC,BD)$ is integrable as well.
\end{proposition}
\begin{proof}
Let $M$ and $N$ be the bidifferential operators corresponding to $A$ and $B$.
By $(3.10)$ for $B$ and $(1.25)$ we can write for all functions $F \in \mc V$ 
\begin{equation}
X_{BD(F)}(B)= (D_{BD})_F B-B(D_D)_F B +B(N)_{D(F)}.
\end{equation}
Multiplying on the right by $D$ we have 
\begin{equation}
X_{BD(F)}(BD)- (D_{BD})_F  BD=B(X_{BD(F)}(D)-(D_D)_F BD+N_{D(F)}D).
\end{equation}
To say that $BD$ is integrable is therefore the same as saying that 
\begin{equation}
Y=X_{BD(F)}(D)-(D_D)_F BD+N_{D(F)}D \hspace{1 mm},
\end{equation}
is divisible on the left by $D$ as a bidifferential operator. By construction of $D$ and by Lemma $1.29$ it will be enough to show that $AY=BP$ for some bidifferential operator $P$. From now on we will work modulo the left ideal of $\mc V[\partial]$ generated by $B$.
\\
\indent
Let us apply the derivation $X_{BD(F)}$ to the identity $AD=BC$. 
\begin{equation}
AX_{BD(F)}(D)+X_{BD(F)}(A)D \equiv X_{BD(F)}(B)C .
\end{equation}
We now want to rewrite the RHS of $(3.34)$. 
Projecting $(3.31)$ modulo $B$, we have  
\begin{equation}
X_{BD(F)}(B) \equiv (D_{BD})_FB .
\end{equation}
Using equations $(3.33)$ to $(3.35)$ we get 
\begin{equation}
\begin{split}
AY \equiv (D_{BD})_F. BC -X_{BD(F)}(A)D&\\
        -A (D_D)_F BD +A(N)_{D(F)} D.
\end{split}
\end{equation}
Replacing $F$ by $D(F)$ in the third line of $(3.19)$ yields
\begin{equation}
X_{A(D(F))}(B)+X_{B(D(F))}(A)-(D_A)_{D(F)} B-(D_B)_{D(F)} A \equiv AN_{D(F)}.
\end{equation}
After comparing $(3.36)$ with $(3.37)$ and recalling that $AD=BC$ we get
\begin{equation}
\begin{split}
AY \equiv (D_{BD})_F AD -A(D_D)_F BD
        + X_{AD(F)}(B)D&\\-(D_B)_{D(F)} AD -(D_A)_{D(F)} BD \hspace{1 mm}.
\end{split}
\end{equation}
Moreover, by $(1.24)$, we have 
\begin{equation}
\begin{split}
(D_{BD})_F AD-(D_B)_{D(F)}AD&=B(D_D)_F AD \equiv 0 \\
A (D_D)_F BD +(D_A)_{D(F)} BD&=(D_{AD})_F BD \hspace{1 mm}.
\end{split}
\end{equation}
Therefore $(3.38)$ can be simplified to 
\begin{equation}
AY \equiv [X_{AD(F)}(B)-(D_{AD})_F B]D.
\end{equation}
After replacing $AD$ by $BC$ and using $(1.24)$ one more time we get 
\begin{equation}
\begin{split}
AY &\equiv [X_{BC(F)}(B)-(D_{BC})_F B]D\\
        & \equiv [X_{BC(F)}(B)-(D_B)_{C(F)} B]D\\
       & \equiv B(N)_{C(F)}D\\
        & \equiv 0.
\end{split}
\end{equation}
\\
We proved that $BD$ is an integrable operator. Since $(B,A)$ is integrable as well by symmetry of $(3.19)$, we infer that $AC$ is an integrable operator as well. Moreover $L=AB^{-1}$ is hereditary by Lemma $3.8$ and hence $L^2=(AC)(BD)^{-1}$ is hereditary too. This enables us to conclude that the pair $(AC,BD)$ is integrable thanks to Lemma $3.9$.
\end{proof}

\begin{proposition}
Let $L$ be an integrable operator. Then $L^n$ is integrable for all $n \in \mathbb{Z}_{+}$
\end{proposition}
\begin{proof}
Let $(A,B)$ be an integrable pair of differential operators such that $L=AB^{-1}$. We know by Lemma $3.8$ that $L$ is hereditary, and therefore so is $L^k$ for all $k \geq 1$ by Lemma $2.17$.
\\
\indent
Let us begin by defining sequences of differential operators $A_n$, $B_n$, $C_n$ and $D_n$. We let $A_1=A$, $B_1=B$, $C_0=1$ and $D_0=1$. Assuming that we defined $A_n$, $B_n$, $C_{n-1}$ and $D_{n-1}$, we let 
\begin{equation}
B_n C_n=AD_n.
\end{equation}
be the least right common multiple of the pair $(A,B_n)$. We then define 
\begin{equation}
A_{n+1}=A_n C_n, \hspace{1 mm} B_{n+1}=BD_n.
\end{equation}
 It is easy to check by induction that $L^n=A_n{B_n}^{-1}$  for all $n \geq 1$. 
\\
\indent
Let us prove by induction that $D_n$ divides $D_{n+1}$ on the left for all $n \geq 1$, i.e. that there exists a differential operator $E_n$ such that $D_{n+1}=D_nE_n$. 
\\
We know that 
\begin{equation}
AD_2=B_2 C_2=B D_1 C_2.
\end{equation}
Hence, since $AD_1=BC_1$ is the right least common multiple of the pair $(A,B)$, we can find an operator $E_2$ such that 
\begin{equation}
D_2=D_1E_2, \hspace{1 mm}D_1 C_2=C_1 E_2.
\end{equation}
Let us rename $D_1$ by $E_1$. Assume that we constructed $n$ operators $E_1, \dots, E_n$ such that for all $1 \leq k \leq n$
\begin{equation}
D_k=E_1 \dots E_k, \hspace{2 mm}
E_kC_{k+1}=C_kE_{k+1}.
\end{equation}
Combining equations $(3.42)$, $(3.43)$ and $(3.46)$, we obtain
\begin{equation}
\begin{split}
AD_{n+1}&=B_{n+1} C_{n+1} \\
             &=B D_n C_{n+1}\\
             &=B D_{n-1} E_n C_{n+1}\\
             &=B_n (E_n C_{n+1}).
\end{split}
\end{equation}
By $(3.42)$, $(3.47)$ and by definition of the least right common multiple, one can find an operator $E_{n+1}$ such that 
\begin{equation}
D_{n+1}=D_nE_{n+1}, \hspace{2 mm}
E_nC_{n+1}=C_nE_{n+1}.
\end{equation}
Therefore $(3.46)$ holds for all $k \geq 1$. Note that $E_k$ and $C_k$ are right coprime for all $k \geq 1$. Indeed, $D_k$ and $C_k$ are right coprime and $E_k$ divides $D_k$ on the right. Moreover, we have for all $k \geq 1$
\begin{equation}
A_k=AC_1 \dots C_{k-1}.
\end{equation}
\\
\indent
  We now prove by induction on $n$ that $(A_{2^n},B_{2^n})$ is integrable for all $n \geq 0$. The case $n=0$ is true by hypothesis. Let us assume that it is true for $k=2^n$ and define operators $K$ and $L$ such that 
\begin{equation}
A_kK=B_kL
\end{equation}
 is the least right common multiple of the pair $(A_k,B_k)$. We are going to show that for some operator $M$, \begin{equation}
A_{2k}M=A_kL, \hspace{2 mm}
B_{2k}M=B_kK.
\end{equation}
Combining $(3.49)$ with $(3.50)$ gives for all $k \geq 1$
\begin{equation}
A(C_1 \dots C_{k-1}K)=B_kL.
\end{equation}
Comparing $(3.42)$ with $(3.52)$ we deduce the existence of an operator $H$ such that
\begin{equation}
L=C_kH, \hspace{2 mm}
C_1 \dots C_{k-1}K=D_kH=E_1 \dots E_kH.
\end{equation}
Using $(3.53)$ and the fact that $C_i.E_{i+1}=E_i.C_{i+1}$ is the least right common multiple of the pair $(E_i,C_i)$ for all $i \geq 1$, we find that there exists a differential operator $\tilde{H}$ such that
\begin{equation}
H=C_{k+1} \dots C_{2k-1} \tilde{H}, \hspace{2 mm}
K=E_{k+1} \dots E_{2k} \tilde{H}.
\end{equation} 
Combining $(3.53)$ and $(3.54)$ gives $(3.51)$ with $M=\tilde{H}$.
By Proposition $3.12$, the pair $(A_k L,B_k K)$ is integrable, hence so is $(A_{2k}, B_{2k})$ by $(3.51)$ and Corollary $3.10$. 
\\
\indent
Since by equations $(3.46)$ and $(3.49)$ $A_m$ (resp. $B_m$)  divides $A_l$ ( resp. $B_l$) on the left whenever $m \geq l$ and all the $A_{2^n}$ ( resp. $B_{2^n}$) are integrable, we know by Lemma $3.7$ that the operators $A_m$ (resp. $B_m$) are integrable for all $m \geq 1$. Moreover, since $L^m=A_m{B_m}^{-1}$ is hereditary for all $m \geq 1$, Lemma $3.9$ enables us to claim that the pairs $(A_m,B_m)$ are integrable for all $m\geq 1$. In particular, $L^m$ is integrable for all $m \in \mathbb{Z}_+$.
\end{proof}

As it was the case with hereditariness, integrability is a necessary condition for a rational operator to produce commuting functions. More specifically, we prove :

\begin{proposition}
Let $L$ be a rational operator where $L=AB^{-1}$ is a minimal fractional decomposition and let $(H_n)_{n \geq 0}$ be a sequence in $\mc V$, which spans an infinite-dimensional space over $\mc C$, such that  
\begin{enumerate}
 \item[(1)]  $B(H_{n+1})=A(H_n)$ for all $n \geq 0$, \\
  \item[(2)]  $\{B(H_n),B(H_m)\}=0$ for all $n,m \geq 0$.
\end{enumerate}
Then the pair $(A,B)$ is integrable.
\end{proposition}
\begin{proof}
Let us begin by defining two bidifferential operators, letting
\begin{equation}
\begin{split}
    M_F:=&X_{A(F)}(A)-(D_A)_FA\\
   N_F:=&X_{B(F)}(B)-(D_B)_F B,
\end{split}
\end{equation}
for $F \in \mc V$. By commutativity of the functions $A(H_n)$, we have for all $n, m \geq 0$ 
\begin{equation}
X_{A(H_n)}(A(H_m))=X_{A(H_m)}(A(H_n)).
\end{equation}
Since $X_{A(H_n)}$ is a derivation of $\mc V$ commuting with $\partial$ and using $(3.55)$ we get
\begin{equation}
\begin{split}
X_{A(H_n)}(A(H_m))&=X_{A(H_n)}(A)(H_m)+A(X_{A(H_n)}(H_m))\\
 &=(M)_{H_n}(H_m)+A(X_{A(H_n)}(H_m))+(D_A)_{H_n}(H_m) .
\end{split}
\end{equation}
On the other hand, by $(1.19)$ and $(1.23)$ we have
\begin{equation}
\begin{split}
X_{A(H_m)}(A(H_n))&=D_{A(H_n)}(A(H_m))\\
 &=(D_A)_{H_n}(H_m)+A(D_{H_n}(A(H_m))).
\end{split}
\end{equation}
Combining equations $(3.56)$, $(3.57)$ and $(3.58)$ we get
\begin{equation}
(M_{H_n}+A(X_{A(H_n)}-D_{H_n}A))(H_m)=0.
\end{equation}
Similarly, replacing $A$ with $B$ and $M$ with $N$, we also have for all $n,m \geq 0$
\begin{equation}
(N_{H_n}+B(X_{B(H_n)}-D_{H_n}B))(H_m)=0.
\end{equation}
Let $CA=DB$ be the least left common multiple of the pair $(A,B)$. Recall that for all $n \geq 0$, $B(H_{n+1})=A(H_n)$. Therefore, multiplying $(3.59)$ on the left by $C$ and $(3.60)$ for $n+1$ on the left by $D$, we obtain for all $n,m \geq 0$ : 
\begin{equation}
  (C(M_{H_n}-AD_{H_n}A))(H_m)=(D(N_{H_{n+1}}-BD_{H_{n+1}}B))(H_m).
\end{equation}
Since this is true for all $m \geq 0$ and the span of the $H_m$ is infinite dimensional over $\mc C$ we have an operator identity  for all $n \geq 0$ :
\begin{equation}
  C(M_{H_n}-AD_{H_n}A)=D(N_{H_{n+1}}-BD_{H_{n+1}}B).
\end{equation}
Since $CA=DB$ is also the right least common multiple of the pair $(C,D)$ because $A.B^{-1}$ is a minimal fractionnal decomposition, we deduce from $(3.62)$ the existence of differential operators $P_n$, such that for all $n \geq 0$,
\begin{equation}
N_{H_{n+1}}-BD_{H_{n+1}}B=BP_n.
\end{equation}
Let $Q$ and $R$ be two bidifferential operators such that for all $F \in \mc V$
\begin{enumerate}
\item[(a)] $N_F=BQ_F+R_F$, \\
\item[(b)] $d(R_F) < d(B)$,
\end{enumerate}
which we can find using Lemma $1.26$. It follows from $(3.63)$ that for all $n \geq 0$, $R_{H_{n+1}}=0$. Since the span of the $H_n$ is infinite dimensional over $\mc C$ we have $R=0$. Therefore $B$ divides $N$ on the left, i.e. $B$ is integrable. Similarly, we can prove that $A$ is integrable. Finally, recall that $L=AB^{-1}$ is hereditary by Proposition $2.17$. Therefore the pair $(A,B)$ is integrable by Lemma $3.9$.
    \end{proof}

The following lemma states that when an integrable operator $L$ is recursion for a function, the latter commutes with its image under $L$. 
\vspace{2mm}
\begin{lemma}
Let $F \in \mc V$ and $(A,B)$ be an integrable pair of differential operators such that $L=AB^{-1}$ is recursion for $B(F)$. Then $A(F)$ and $B(F)$ commute.
\end{lemma}
\begin{proof}
We know that $B$ is integrable, meaning  by $(3.10)$ that there exists a bidifferential operator $M$, such that  for all $G \in \mc V$,
\begin{equation}
X_{B(G)}(B)=D_{B(G)}B+B(M_G-D_G B).
\end{equation}
Since $M$ is skewsymmetric, we have 
\begin{equation}
M_G(G)=0  \hspace{2 mm} \forall G \in \mc V.
\end{equation}
Let us now recall that $L$ being recursion for $B(F)$ means
   \begin{equation}
X_{B(F)}(AB^{-1})=[D_{B(F)},AB^{-1}].
\end{equation}
Rearranging the terms of $(3.66)$ and multiplying on the right by $B$, we get 
 \begin{equation}
X_{B(F)}(A)-D_{B(F)}A=AB^{-1}(X_{B(F)}(B)-D_{B(F)}B).
\end{equation}
Let $CA=DB$ be the left common multiple of the pair $(A,B)$. We deduce from $(3.67)$ that
\begin{equation}
C(X_{B(F)}(A)-D_{B(F)}A)=D(X_{B(F)}(B)-D_{B(F)}B).
\end{equation}
Let $CA_0=DB_0$ be the right least common multiple of the pair $(C,D)$. Let $E$ be such that $A=A_0E$ and $B=B_0E$.  From $(3.68)$ it follows that there exists a differential operator $H$ such that
\begin{equation}
     X_{B(F)}(A)-D_{B(F)}A = A_0H, \hspace{2 mm}
     X_{B(F)}(B)-D_{B(F)}B = B_0H .
\end{equation}
Comparing $(3.64)$ with the second line of $(3.69)$, we have
\begin{equation}
H=E(M_F-D_FB).
\end{equation} 
Therefore, by $(3.65)$ and $(3.70)$,
\begin{equation}
H(F)=-E(D_F(B(F))).
\end{equation}
 Applying the first line of $(3.69)$ to $F$ and using $(3.71)$ we get 
\begin{equation}
X_{B(F)}(A(F))=D_{B(F)}(A(F)),
\end{equation}
proving the claim.
\end{proof}

The next proposition provides us with the first  sufficient condition for a rational operator $L$ to generate an infinte sequence of commuting functions, providing that $L^n F_0$ is defined for all $n$ and some $F_0$.
\vspace{4 mm}
\begin{proposition}
Let $L=AB^{-1}$ be a rational operator with $(A,B)$ integrable and let $(H_n)_{n \geq 0}$ be a sequence in $\mc V$ such that
\begin{enumerate}
\item[(1)] $L$ is recursion for $B(H_0)$, \\
\item[(2)] $A(H_n)=B(H_{n+1})$ for all $n \geq 0$.
\end{enumerate}
 Then the functions $B(H_n)$ pairwise commute.
\end{proposition}
\begin{proof}
$L$ is in particular hereditary by Lemma $3.8$. Therefore, since $L$ is recursion for $B(G_0)$ and by Lemma $2.15$, $L$ is recursion for all the $B(H_n)$. We know that $L^k$ is integrable for all $k \geq 0$ by Proposition $3.13$. Let $L^k=A_k B_k^{-1}$ be the right minimal fractional decomposition of $L^k$. By Theorem $4.12$ in [CDSK14] we know that for any $n \geq 0$ there exists a function $F_{n,k}$ such that $B(H_n)=B_k(F_{n,k})$ and $B(H_{n+k})=A_k(F_{n+k})$. It follows directly from Lemma $3.15$ that $B(H_n)$ and $B(H_{n+k})$ commute in $\mc V$. This holds for all $n,k \geq 0$. 
\end{proof}

\begin{remark}
The first condition in Proposition $3.16$ is met whenever $G_0$ lies in the kernel of $B$.
\end{remark}

\begin{corollary}
Let $A$, $B$ be two differential operators and let $(H_n)_{n \geq 0}$ be a sequence of functions in $\mc V$ which spans an infinite-dimensional space over $\mc C$. Let us assume moreover that $L$ is recursion for $B(H_0)$ and that for all $n \geq 0$,
\begin{equation}
B(H_{n+1})=A(H_n).
\end{equation}
 Then the functions $B(H_n)$ pairwise commute if and only if the pair $(A,B)$ is integrable.
\end{corollary}
\begin{proof}
It follows immediately from Propositions $3.14$ and $3.16$.
\end{proof}

\begin{example}
The rational operator $L=\partial(\partial+u)\partial^{-1}$ is integrable. It follows from Example $3.3$ after performing the change of variables $u \rightarrow u + \lambda$. On the other hand, if we let $A=\partial(\partial+u)$, $B=\partial$, and $H_n=(\partial+u)^n(1)$, it is clear that $L$ is recursion for $B(H_0)=0$ and that $B(H_{n+1})=A(H_n)$ for all $n \geq 0$. Therefore the functions $H_n'$ pairwise commute. Note that $u_t=H_2'$ is the Burgers equation.
\end{example}
The following proposition says that hereditary operators are not far from being integrable.
\begin{proposition}
Let $L$ be a hereditary rational operator with right minimal fractional decomposition $AB^{-1}$ and left minimal fractional decomposition $C^{-1}D$. Assume moreover that $(C,D)$ are right coprime. Then the pair $(A,B)$ is integrable.
\end{proposition}
\begin{proof}
Recall that thanks to Lemma $3.8$ it is enough to show that $B$ is integrable. If we look carefully at the proof of Lemma $3.8$, more specifically if we equate the second and the fifth line of $(3.25)$, we see that the hereditariness of $L$ is equivalent to the equation 
\begin{equation}
P_F=LQ_F-L^2 R_F  \hspace{2 mm} \text{     for all } F \in \mc V,
\end{equation}
 where the bidifferential operators $P,Q,R$ are given by
\begin{equation}
\begin{split}
P_F&=X_{A(F)}(A)-(D_A)_F A \\
Q_F&=X_{A(F)}(B)+X_{B(F)}(A)-(D_A)_F B-(D_B)_F A \\
R_F&=X_{B(F)}(B)-(D_B)_F B.
\end{split}
\end{equation}
Asking for the integrability of $B$ amounts to ask for $R$ to be divisible on the left by $B$.  We will show that fact. Let us first use the left presentation $L=C^{-1}D$ in $(3.74)$ :
\begin{equation}
  P_F=C^{-1}DQ_F+C^{-1}DC^{-1}DR_F.
\end{equation}
 Rearranging $(3.76)$, we get for all $F \in \mc V$
\begin{equation}
 CP_F-DQ_F=DC^{-1}DR_F.
\end{equation}
 Now we use the right minimality of the fraction $DC^{-1}$ to deduce that $DR$ is divisible on the left by $C$. Given that $CA=DB$ is the least right common multiple of the pair $(C,D)$ since $AB^{-1}$ is a right minimal decomposition of $L$, it follows from Lemma $1.29$ that $R$ is divisible on the left by $B$.
\end{proof}

\begin{example}
In Example $3.11$, we saw that the rational operator $L=\partial^{-1}u''\partial$ is not integrable. Note that $L$ does not meet the hypothesis of Proposition $3.20$. Indeed, $\partial$ and $u'' \partial$ are not right coprime.
\end{example}

\section{Weakly non-local Operators}
In this section we study weakly non-local rational operators. Let $\mc V$ be a normal algebra of differential functions, and let $\mc K$ be its field of fractions.

\begin{definition}
A weakly non-local operator $L$ is a rational operator which can be written in the following form 
\begin{equation}
L=E(\partial)+ \sum_{i=1}^n{p_i {\partial}^{-1}q_i},
\end{equation}
where $E$ is a differential operator and $p_i$ and $q_i$ are elements of $\mc V$.  We denote the space of weakly non-local operators with coefficients in $\mc V$ by $W_{\mc V}$, or simply $W$ when there is no confusion on the algebra $\mc V$.
\end{definition}

\begin{definition}
Let $\mc A$ be a differential algebra with the subfield of constants $\mc C$ and let $P \in \mc A[\partial]$.
We say that $P$ has a full kernel in $\mc A$ if 
\begin{equation}
dim_{\mathcal{C}}\hspace{1 mm} Ker_{\mc A} P=d(P).
\end{equation}
\end{definition}

As we show next, weakly non-local operators can be characterized by their denominators in a minimal fractional decomposition $AB^{-1}$.

\begin{lemma}
Let $L=E(\partial)+ \sum_{i=1}^n{p_i {\partial}^{-1}q_i}$ be a weakly non-local operator, and $A$ and $B$ be two differential operators. Then, we have :
\begin{equation}
\begin{split}
    AL&=\sum_{i=1}^n{A(p_i) {\partial}^{-1}q_i}  \hspace{2 mm} mod \hspace{2 mm} \mc V[\partial], \\
    LB&=\sum_{i=1}^n{p_i {\partial}^{-1}B^*(q_i)} \hspace{2 mm} mod \hspace{2 mm} \mc V[\partial].
\end{split}
\end{equation}
\end{lemma}
\begin{proof}
This follows directly from the fact that, for all $i=1,\dots,n$ we can find by the Euclidean division two differential operators $C_i$ and $D_i$ such that
\begin{equation}
Ap_i=C_i\partial+A(p_i), \hspace{2 mm}
q_iB=\partial D_i +B^*(q_i).
\end{equation}
\end{proof}

\begin{lemma}
The vector space $W_{\mc V}$ is isomorphic to $\mc V [\partial] \oplus ({\mc V} \otimes _ {\mc C} \mc V)$ under the map $E+\sum_i{p_i{\partial}^{-1}q_i} \rightarrow E \oplus (\sum_i{p_i \otimes q_i})$.
\end{lemma}

\begin{proof}
Let $Z=span \{p{\partial}^{-1}q | (p,q) \in {\mc V}^2 \}$. We need to prove that $Z$ is isomorphic to ${\mc V} \otimes _ {\mc C} \mc V$. Let $\phi$ be the map from ${\mc V} \otimes _ {\mc C} \mc V$ to $Z$ sending the tensor $f \otimes g $ to the operator $f{\partial}^{-1} g$. $\phi$ is surjective by definition of $Z$. To check that it is an injective map let us take two sets consisting of linearly independent functions $\{f_1,\dots,f_n\}$ and $\{g_1,\dots,g_n\}$ and assume that 
\begin{equation}
L=\sum_{i=1}^n {{f_i}{\partial}^{-1}{g_i}}=0.
\end{equation}
 If we expand $L$ as a Laurent series in $\partial$ and equal its coefficients to $0$ we get that for all nonnegative integer $k$, $\sum_{i=1}^n {{f_i}{g_i}^{(k)}}=0$ which is the same as saying that for all differential operator $P$, $\sum_{i} {{f_i}P(g_i)}=0$. Since the $g_j$ are linearly independent functions and for a given $i$, we can pick by Lemma $1.17$ an operator $P_i$ annihilating all the $g_j$'s except $g_i$. Therefore $f_i$ should be trivial, which is a contradiction.
\end{proof}

\begin{lemma}
Let $L \in \mc K(\partial)$ be a rational operator with minimal right fractional decomposition $L=AB^{-1}$. Then the following statements are equivalent
\begin{enumerate}
\item[(1)] $L \in W_{\mc K}$.\\
\item[(2)] $B$ has a full kernel in $\mc K$. \\
\item[(3)] $B^*$ has a full kernel in $\mc K$.
\end{enumerate}
Moreover if $L=E(\partial)+\sum_{i=1}^n {{p_i}{\partial}^{-1}{q_i}}$ where both the $p_i$'s and $q_i$'s are linearly independent elements of $\mc K$, then $d(B)=n$ and $B$ is a right least common multiple of the differential operators $\frac{1}{q_i}\partial$. Finally, $Ker B^*$ is spanned by the $q_i$'s.
\end{lemma}
\begin{proof}
Let us first prove that $(1) \implies (3)$. Let $E \in \mc K[\partial]$ and $p_i,q_i \in \mc K$ be such that
\begin{equation}
L=AB^{-1}=E+\sum_{i=1}^n {{p_i}{\partial}^{-1}{q_i}},
\end{equation}
 where the $\{p_1, \dots, p_n\}$ and $\{q_1,\dots,q_n\}$ are sets consisting of linearly independent functions. Then, multiplying $(4.6)$ on the right by $B$ and using Lemma $4.3$ we obtain 
\begin{equation}
\sum_{i} {{p_i}{\partial}^{-1}{B^{*}(q_i)}}=0.
\end{equation} 
By Lemma $4.4$, $(4.7)$ implies that $\{q_1, \dots, q_n\} \subset Ker B^*$. Hence 
\begin{equation}
n \leq dim_{\mc C} Ker B^* .
\end{equation}
Let $C$ be a common multiple of the differential operators $\frac{1}{q_i}\partial$, i.e. a differential operator such that for all $i=1,\dots,n$ we can find a differential operator $M_i$ satisfying 
\begin{equation}
C=\frac{1}{q_i}\partial M_i.
\end{equation}
From equations $(4.7)$ and $(4.9)$ we have
\begin{equation}
\begin{split}
AB^{-1}&=E+\sum_{i=1}^n{p_i(\frac{1}{q_i} \partial)^{-1}} \\
             &=E+\sum_{i=1}^n{p_iM_iM_i^{-1}(\frac{1}{q_i} \partial)^{-1}}\\
             &= (EC+\sum_{i=1}^n{p_iM_i})C^{-1}.
\end{split}
\end{equation}
 Since $AB^{-1}$ is a right minimal fractional decomposition of $L$, there exists by Lemma $1.16$ an operator $D$ such that 
\begin{equation}
\begin{split}
A&=(EC+\sum_{i=1}^n{p_iM_i})D \\
B&=CD.
\end{split}
\end{equation}
In particular,
\begin{equation}
d(B) \leq d(C)=n.
\end{equation}
Recall that $d(B)=d(B^*)$. Hence, by $(4.8)$, $(4.12)$ and Lemma $1.15$, 
\begin{equation}
d(B^*)=dim_{\mc C} Ker B^*=n,
\end{equation}
 which means that $B^*$ has a full kernel in $\mc K$ spanned by the functions $q_i$. We also get from $d(B)=d(C)$ that $B$ is a least right common multiple of the differential operators $\frac{1}{q_i}\partial$, and that
\begin{equation}
A=EB+\sum_{i=1}^n{p_iM_i}.
\end{equation}
\\
\indent
We now prove that $(2) \implies (1)$ by induction  on the degree of $B$. We only have to prove that $B^{-1} \in W_{\mc K}$ since it follows from Lemma $4.3$ that $W_{\mc K}$ is stable by left or right multiplication by elements of $\mc V[\partial]$.
 If $B$ is a degree $1$ operator such that $dim_{\mc C} Ker B =1$ then it can be written in the form $f \partial g$. Therefore $B^{-1}=\frac{1}{g}\partial^{-1}\frac{1}{f}$ is weakly nonlocal. Let $B$ be an operator of degree $n+1$ with full kernel in $\mc K$ and $f$ an element in its kernel. Then we can find an operator $E$ such that 
\begin{equation}
B=E\partial \frac{1}{f}.
\end{equation}
The map 
\begin{equation}
\begin{split}
  \Phi :& Ker B \to Ker E\\
  & g \mapsto (g/f)'.
\end{split}
\end{equation}
has a one dimensional kernel spanned by $f$. Since by hypothesis $Ker B$ is $(n+1)-$dimensional, we get that $Ker E$ is at least $n$-dimensional. We also know that $d(E)=n$. Therefore by Lemma $1.15$, $E$ has a full kernel in $\mc K$. Moreover, $\Phi$ is surjective, i.e.
\begin{equation}
Ker E=Im \Phi \subset \partial K.
\end{equation}
 By the induction hypothesis, $E^{-1}$ is weakly non-local. Let $\{a_i\}$ and $\{b_i\}$ be two sets of linearly independent functions such that
\begin{equation}
E^{-1}=\sum_{i=1}^n{a_i \partial^{-1} b_i}.
\end{equation}
Multiplying $(4.18)$ on the left by $E$ we have from Lemma $4.3$ :
\begin{equation}
\sum_{i=1}^n{E(a_i )\partial^{-1} b_i}=0.
\end{equation}
Therefore, by Lemma $4.4$, $E(a_i)=0$ for $i=1,\dots,n$. Moreover, we know that $Ker E$ is a subset of $\partial \mc K$ thanks to $(4.17)$. For $i=1 \dots n$, let $d_i \in \mc K$ be such that $a_i=d_i'$. Then, we have
\begin{equation}
\begin{split}
B^{-1}&=f\partial^{-1}E^{-1} \\
     &=\sum_{i=1}^n{ f\partial^{-1} d_i' \partial^{-1}b_i }\\
     &=\sum_{i=1}^n{ f d_i \partial^{-1} b_i-f\partial^{-1}d_i b_i}.
\end{split}
\end{equation}
where we used the identity $\partial h-h\partial=h'$, valid for all $h \in \mc K$. Therefore, $B^{-1}$ is weakly non-local.
\\
\indent
Finally, let us check that $(3)\implies (2)$. From the second step of the proof, we know that ${B^*}^{-1}$ is weakly non-local. Using $(1) \implies (3)$, we deduce that $B$ has a full kernel in $\mc K$.
\end{proof}

\begin{remark}
It follows from the proof of Lemma $4.5$ that a rational operator $L \in \mc V(\partial)$ lies in $W_{\mc V}$ if and only if both $B^*$ and $C$ have a full kernel in $\mc V$, where $AB^{-1}$ (resp. $C^{-1}D$) is a right (resp. left) minimal fractional decomposition of $L$.
\end{remark}
Full kernel operators admit an interesting characterization in terms of integrability.

\begin{proposition}
Let $B \in \mc V[\partial]$ be a differential operator with full kernel in $\mc V$. Then it is integrable if and only if $Ker B^{*}$ is spanned
by variational derivatives.
\end{proposition}

\begin{proof}
We keep the notations of Lemma $4.5$. Namely : $Ker B^*$ is spanned by the linearly independent elements $q_i$, $i=1,\dots,n$ of $\mc V$ and $B=\frac{1}{q_i}\partial M_i$ for all $i$ for some differential operators $M_i \in \mc K[\partial]$. By $(3.10)$, $B$ is integrable if and only if $B$ divides on the left the differential operator $X_{B(F)}(B)-D_{B(F)}B$ for all $F \in \mc V$. Since $B$ is a right least common multiple of the operators $\frac{1}{q_i} \partial$, it is equivalent to say that for all $i=1,\dots,n$ and for all $F \in \mc V$,  $X_{B(F)}(B)-(D_B)_F B$ is divisible on the left by $\frac{1}{q_i} \partial$ which itself amounts to say that for all $F \in \mc V$ and $i=1,\dots,n$
\begin{equation}
X_{B(F)}(B^*)(q_i)=B^*({D_{B(F)}}^*(q_i)).
\end{equation}
As we have $B^*(q_i)=0$ by definition, the LHS of $(4.21)$ can be rewritten as $-B^*(X_{B(F)}(q_i))$. Taking the adjoint of $(1.24)$ and applying to $q_i$ yields
\begin{equation}
{D_{B(F)}}^*(q_i)=(D_B)_F^*(q_i).
\end{equation}
Both $X_{B(F)}(q_i)$ and $(D_B)_F^*(q_i)$ are differential operators applied to $F$. Therefore the integrability of $B$ is equivalent to the following set of equations for $i=1 \dots n$ and all $F \in \mc V$
\begin{equation}
 X_{B(F)}(q_i)=-{D_{B(F)}}^*(q_i) .
\end{equation}
 The RHS of $(4.23)$ can be rewritten as follows :
\begin{equation}
\begin{split} 
{D_{B(F)}}^*(q_i)&={D_{\frac{M_i(F)'}{q_i}}}^*(q_i)\\
                       &= (M_i(F)' D_{\frac{1}{q_i}}+\frac{1}{q_i}\partial D_{M_i(F)})^*(q_i) \\
                       &= (-\frac{1}{q_i^2} M_i(F)' D_{q_i})^*(q_i) \\
                       &= -D_{q_i}^*(B(F)).
\end{split}
\end{equation}
Plugging $(4.24)$ in $(4.23)$, we see that $B$ is integrable if and only if for all $F \in \mc V$ and $i=1,\dots,n$, we have
\begin{equation}
D_{q_i}(B(F))=D_{q_i}^*(B(F)).
\end{equation}
Since the image of $B$ is infinite-dimensional over $\mc C$, we can simplify $(4.25)$ into
\begin{equation}
D_{q_i}=D_{q_i}^*, \hspace{2 mm} i=1 \dots n.
\end{equation}
Using Lemma $1.23$, we conclude that $B$ is integrable if and only if 
\begin{equation}
Ker B^* \subset \frac{\delta \mc V}{\delta u}.
\end{equation}
\end{proof}

In the next lemma we give some technical yet useful result on the vector space structure of some extension of the space of weakly non-local operators.

\begin{lemma}
Let $U$ be the following subspace of rational operators 
\begin{equation}
U=\{ E(\partial)+\sum {p_{\alpha}\partial^{-1} q_{\alpha}}+\sum {a_{\beta}\partial^{-1}b_{\beta}\partial^{-1}c_{\beta}}\}.
\end{equation}
Then
\begin{equation}
 {\mc V} \otimes_{\mc C} {\mc V /{\partial \mc V}} \otimes_{\mc C} \mc V \simeq U/W_{\mc V}
\end{equation}
via the morphism $\phi$ sending $a \otimes \int{b} \otimes c$ to $a \partial^{-1} b \partial^{-1} c$.
\end{lemma}

\begin{proof}
Let $\phi$ be the map from ${\mc V} \otimes _ {\mc C} {{\mc V}/{\partial \mc V}} \otimes _{\mc C} \mc V$ to $U/W_{\mc V}$ sending the tensor $f \otimes \int{g} \otimes h$ to the image of the nonlocal operator $f{\partial}^{-1}{g}{\partial}^{-1}h$ in $U/W_{\mc V}$. It is well defined because $f\partial^{-1}g'\partial^{-1}h=0$ in $U/W_{\mc V}$ for any triple $(f,g,h)$ thanks to the identity $g'=\partial g-g \partial$. It is surjective since any class in $U/W_{\mc V}$ contains a sum of elements of the form $f\partial^{-1}g\partial^{-1}h$. 
\\
It remains to check the injectivity of $\phi$. We will do so by proving by induction on $n$ the following statement :
\\
\hspace{4 mm} If $\{f_1,\dots,f_n \}$ and $\{\int{g_1},\dots,\int{g_m} \}$ are two sets of linearly independent elements in $\mc V$ and ${{\mc V}/{\partial \mc V}}$, and $\{h_{ij}\}$ are functions in $\mc V$ then $L=\sum_{i,j}{{f_i}{\partial}^{-1}{g_j}{\partial}^{-1}{h_{ij}}} \in W_{\mc V} $ if and only if $h_{ij}=0$ for all $(i,j)$.
\\
 If $n=1$, we have to see why $L=\sum_{j=1}^m{f{\partial}^{-1}g_j {\partial}^{-1}h_j}\in W$ implies $h_j=0$ for $j=1,\dots,m$. Let $a_l,b_l \in \mc V$ such that
\begin{equation}
\sum_{j=1}^m{f{\partial}^{-1}g_j {\partial}^{-1}h_j}=\sum_{l=1}^k{a_l \partial^{-1} b_l}.
\end{equation}
Multiplying on the left by $\partial \frac{1}{f}$ we get, thanks to Lemma $4.3$,
\begin{equation}
\sum_{j=1}^m { g_j {\partial}^{-1} h_j}=\sum_{l=1}^k {{(a_l/f)}' {\partial}^{-1} b_l}.
\end{equation}
Let $e_1,\dots,e_s$ be a basis over $\mc C$ of the space spanned by the $h_j$'s and $b_l$'s. Let $\alpha_{jr} \in \mc C$ (resp. $\beta_{lr}$) be the coordinates of $h_j$ (resp. $b_l$) in this basis. Then equation $(4.31)$ rewrites into 
\begin{equation}
\sum_{r=1}^s{(\sum_{j=1}^m \alpha_{jr}g_j-\sum_{l=1}^k{\beta_{lr}(a_l/f)'})\partial^{-1} e_r}=0.
\end{equation}
From Lemma $4.4$ we conclude that, for $r=1,\dots,s$ :
\begin{equation}
\sum_{j=1}^m \alpha_{jr}g_j=\sum_{l=1}^k{\beta_{lr}(a_l/f)'}.
\end{equation}
Projecting $(4.33)$ into $\mc V/ \partial \mc V$, we get
\begin{equation}
\sum_{j=1}^m \alpha_{jr}\int g_j=0.
\end{equation}
Since we assume the functionals $\int g_j$ to be linearly independent, we get that all the coordinates $\alpha_{jr}$ must be $0$, hence that $h_j=0$ for all $j$, proving the statement for $n=1$.
Assume that our statement holds for $n-1 \geq 1$ and let $L$ be such that  
\begin{equation}
L=\sum_{i=1}^n{\sum_{j=1}^m{{{f_i}{\partial}^{-1}{g_j}{\partial}^{-1}{h_{ij}}}}} \in W_{\mc V},
\end{equation} 
where $\{f_1,\dots,f_n \}$ and $\{\int{g_1},\dots,\int{g_m} \}$ are two sets of linearly independent elements in $\mc V$ and ${{\mc V}/{\partial \mc V}}$.
Let $i_0 \in [1,n]$. Multiplying on the left equation $(4.35)$ by $\partial \frac{1}{f_{i_0}}$ we have 
\begin{equation}
\sum_{i \neq i_0,j}{{(\frac{f_i}{f_{i_0}})'}{\partial}^{-1}{g_j}{\partial}^{-1}{h_{ij}}} \in W_{\mc K} .
\end{equation}
  Given that the functions $(\frac{f_i}{f_{i_0}})'$ are linearly independent over $\mc C$, we use the induction hypothesis to deduce that 
$h_{ij}=0$ for $i \neq 0$ and $j=1,\dots,m$. Since this is true for all $i_0=1,...,n$, we conclude that $h_{ij}=0$ for all $i=1,...,n$ and all $j=1,...,m$.
\end{proof}

\begin{lemma}
Let $L=E+\sum_{i=1}^n{p_i{\partial}^{-1}q_i}$ be a weekly non-local operator and $AB^{-1}$ be its right minimal fractional decomposition, with $d(B)=n$. Then the following are equivalent 
\begin{equation}
 \begin{split}
    (1) & \hspace{2 mm} A(Ker B) \subset Im B.\\
    (2) &\hspace{2 mm} p_iq_j \in \partial \mathcal{V} \hspace{2 mm} \forall (i,j) \in \{1,n\}^2.\\
     (3) &\hspace{2 mm} L^2  \text{ is weakly non-local}.
\end{split}
\end{equation} 
\end{lemma}   
\begin{proof}
By Lemma $4.3$, we see that
\begin{equation}
L^2= \sum_{i=1}^n{\sum_{j=1}^m{p_i\partial^{-1}q_ip_j \partial^{-1} q_j}} \hspace{2 mm} mod \hspace{2 mm} W.
\end{equation}
Therefore $(2) \iff (3)$ is a direct application of Lemma $4.8$. $B$ is a right least common multiple of the operators $\frac{1}{q_i} \partial$, which implies that $p \in Im B$  if and only if $q_ip$ is a total derivative for all $i$. In other words, $(2)$ is saying that $p_i \in Im B$ for all $i$. To prove that $(1) \iff (2)$, let us check that $A(Ker B)= \langle p_1, \dots, p_n \rangle$. Recall that 
\begin{equation}
\begin{split}
 B&=\frac{1}{q_i} \partial M_i, \hspace{2 mm} i=1,\dots,n, \\
 A&= EB+\sum_{i=1}^n{p_iM_i}.
\end{split}
\end{equation}
From the first line of $(4.39)$, we get that $M_i(x)$ is a constant for all $x \in Ker B$ and all $i$. Using the second line of $(4.39)$, we see that $A(Ker B) \subset \langle p_1,\dots,p_n \rangle$. Since $A$ and $B$ are right coprime, $Ker A \cap Ker B= \{0\}$ and $dim A(Ker B)=dim Ker B=n$. Therefore, $A(Ker B)= \langle p_1, \dots, p_n \rangle$.
\end{proof}

In the two next propositions, we examine what can be said of rational operators which are both integrable and weakly non-local.  

\begin{proposition}
Let $L=E+\sum_{i=1}^n{p_i{\partial}^{-1}q_i}$ be an integrable weakly non-local operator where $\{p_i\}$ and $\{q_i\}$ are sets consisting of linearly independent functions over $\mc C$. Then the space $V= \langle q_1,\dots,q_n \rangle$ is spanned by variational derivatives and $\frac{\delta }{\delta u} (p_iq_j)\in V$ for all $(i,j) \in \{1,n\}^2$ 
\end{proposition}

\begin{proof}
We already know from Lemma $4.5$ that the kernel of $B^*$ is spanned by the $q_i$'s where $AB^{-1}$ is a right minimal fractional decomposition of $L$. Proposition $4.7$ says that the $q_i$'s must be variational derivatives. Since $L$ is integrable it is hereditary in particular (Lemma $3.8$). In particular, $AB^{-1}$ satisfies equation $(2.37)$. 
The LHS of $(2.37)$ lies in $W$ for all $F \in \mc V$. Indeed, $X_{A(F)}$ preserves $W$ and so does multiplication by a differential operator. Therefore, the RHS of $(2.37)$ lies in $W$ for all $F \in \mc V$ :
\begin{equation}
(E+\sum_{i=1}^n{p_i{\partial}^{-1}q_i})( \mc L_{B(F)}(E)+\mc L_{B(F)}(\sum_{j=1}^n{p_j \partial^{-1} q_j})) \in W_{\mc V}.
\end{equation}
For the same reasons as above, we have for all $F \in \mc V$ :
\begin{equation}
\sum_{i=1}^n{p_i \partial^{-1} q_i \mc L_{B(F)}(\sum_{j=1}^n{p_j \partial^{-1} q_j})} \in W_{\mc V}.
\end{equation}
In other words, for all $F \in \mc V$ we have 
\begin{equation}
\sum_{i,j} p_i \partial ^{-1} q_i (X_{B(F)}(p_j \partial ^{-1} q_j ) - [ D_{B(F)}, p_j \partial ^{-1} q_j ]) \in W_{\mc V}.
\end{equation}
By Lemma $4.3$, we simplify $(4.43)$ into, for all $F \in \mc V$,
\begin{equation}
\begin{split}
&\sum_{i,j}{p_i\partial^{-1}q_ip_j \partial^{-1}(X_{B(F)}(q_j)+D_{B(F)}^*(q_j))} \\
+& \sum_{i,j}{p_i \partial^{-1}q_i(X_{B(F)}(p_j)-D_{B(F)}(p_j))\partial^{-1}q_j} \in W_{\mc V}.
\end{split}
\end{equation}
Recall by $(4.24)$ and $(4.26)$ that for all $F \in \mc V$ and for all $i=1,\dots,n$ :
\begin{equation}
D_{B(F)}^*(q_i)=-D_{q_i}^*(B(F))=-D_{q_i}(B(F))=-X_{B(F)}(q_i).
\end{equation}
Hence the first term in $(4.43)$ vanishes and we get that for all $F \in \mc V$,
\begin{equation}
\sum_{i,j}{p_i \partial^{-1}q_i(X_{B(F)}(p_j)-D_{B(F)}(p_j))\partial^{-1}q_j} \in W_{\mc V}.
\end{equation}
Since both the $p_i$'s and the $q_i$'s are linearly independent, it follows from Lemma $4.8$ and equation $(4.45)$ that for all $F \in \mc V$ :
\begin{equation}
 q_i X_{B(F)}(p_j)-q_i D_{B(F)} ( p_j) \in \partial \mc V.
\end{equation}
Let us work in the quotient space $\mc V / \partial \mc V$. For all $F \in \mc V$
\begin{equation}
q_i D_{p_j}(B(F)) \equiv D_{B(F)}(p_j)q_i.
\end{equation}
By $(1.31)$ and $(4.44)$, $(4.47)$ simplifies into
\begin{equation}
B(F)D_{p_j}^*(q_i) \equiv p_j D_{B(F)}^*(q_i) \equiv -p_j D_{q_i}(B(F)) \equiv -B(F) D_{q_i}^*(p_j).
\end{equation}
Hence, for all $F \in \mc V$ : 
\begin{equation}
B^*(D_{p_j}^*(q_i)+D_{q_i}^*(p_j))F \equiv 0.
\end{equation}
Since $(4.49)$ holds for all $F \in \mc V$ and that $ \partial \mc V \neq \mc V$ ([BDSK09]), we have 
\begin{equation}
B^*(D_{p_j}^*(q_i)+D_{q_i}^*(p_j))=B^*(\frac{\delta }{\delta u} (p_iq_j))=0.
\end{equation}
  Therefore, $\frac{\delta }{\delta u} (p_iq_j)\in V$ for all $(i,j) \in [1,n]^2$.
\end{proof}

For a weakly-non-local operator with minimal fractional decomposition $L=AB^{-1}$, we only need to check that $L$ is hereditary and that the functions appearing on the right of the reduced non-local part of $L$ are variational derivatives to claim that $L$ is integrable.

\begin{proposition}
Let $L=E+\sum_{i=1}^n{p_i \partial^{-1} q_i}$ be a weakly non-local rational operator, where the $p_i$'s and the $q_i$'s are linearly independent over $\mc C$. Then $L$ is integrable if and only if the functions $q_i$ are variational derivatives
and $L$ is hereditary.
\end{proposition}
\begin{proof}
Let $AB^{-1}$ be a minimal fractional decomposition of $L$. By corollary $3.10$, $L$ is integrable if and only if the pair $(A,B)$ is integrable.
Moreover, by Lemma $3.8$, the pair $(A,B)$ is integrable if and only if $L$ is hereditary and $B$ is integrable. Finally, by Proposition $4.7$, $B$ is integrable if and only if the functions $q_i$ are variational derivatives.
\end{proof}

\section{A sufficient condition of integrability}
In this section we will prove that integrability is a necessary and sufficient condition to generate infinitely many commuting functions for a weakly non-local operator "preserving" a certain decomposition of $\mc V$.
\\
\indent
 Let $\sigma$ be an involution of $\mathcal{V}$. Let 
 \begin{equation}
    \mathcal{V}=\mathcal{V}_{\bar{0}} \oplus \mathcal{V}_{\bar{1}}
\end{equation}
be the eigenspace decomposition of $\mc V$ for $\sigma$. We call elements of $\mc V_{\bar{0}}$ even functions and elements of $\mc V_{\bar{1}}$ odd functions. Note that, since $\sigma$ is an algebra morphism, we have for all $\bar{i},\bar{j} \in \mathbb{Z}/2\mathbb{Z}$, $\mc V_{\bar{i}}. \mc V_{\bar{j}} \subset \mc V_{\bar{i}+\bar{j}}$. Assume furthermore that $\sigma \partial \sigma^{-1}=-\partial$ and that for all $n \geq 0$, $\sigma \frac{\partial}{\partial u^{(n)}} \sigma^{-1}=(-1)^n\frac{\partial}{\partial u^{(n)}}$. In other words, $\partial$ switches parity and $\frac{\delta}{\delta u}$ preserves it.

\begin{example}
The algebra of differential polynomials $R$ admits such a decomposition by declaring $u$ to be even and $\partial$ to be odd. However, it is not the only way to decompose this algebra with the constraints we just defined. 
\end{example}

As usual, an endomorphism of $\mc V$ is called even if it preserves the decomposition $(5.1)$, and odd if it switches parity. Here is the decomposition of $\mc V[\partial]$ into even and odd parts :
\begin{equation}
  \mathcal{V}[\partial]=(\mathcal{V}_{\bar{0}}[{\partial}^2]+\mathcal{V}_{\bar{1}}[\partial^2]\partial) \oplus (\mathcal{V}_{\bar{1}}[{\partial}^2]+\mathcal{V}_{\bar{0}}[\partial^2]\partial).
\end{equation}

 We extend this decomposition to pseudodifferential operators by declaring that $\partial^{-1}$ is odd.

\begin{lemma}
Let $L=E(\partial)+\sum_{i=1}^n{p_i{\partial}^{-1}q_i}$ be a weakly non-local integrable rational operator, where $E$ is an even differential operator, $q_i$'s are linearly independent elements of $\mc V_{\bar{0}}$ and $p_i$'s are linearly independent elements of $\mc V_{\bar{1}}$. Let $p \in \mc V_{\bar{1}}$ be such that $L$ is recursion for $p$. Finally, let $B$ be a right least common multiple of the operators $\frac{1}{q_i} \partial$ and $A$ be such that $L=AB^{-1}$. Then the following statements hold
\begin{enumerate}
\item[(1)] $p$ lies in the image of $B$.\\
\item[(2)] $\{p,p_i\}=0$ for $i=1 \dots n$. \\
\item[(3)] $pq_i$ is a total derivative for $i=1 \dots n$.
\end{enumerate}
Moreover, if $F \in \mc V$ is such that $B(F)=p$, then $A(F) \in \mc V_{\bar{1}}$ and $L$ is recursion for $A(F)$.
\end{lemma}
\begin{proof}
By Proposition $4.11$, the functions $q_i$ are variational derivatives. In particular, for $i=1,...,n$, we have 
\begin{equation}
D_{q_i}=D_{q_i}^*.
\end{equation}
Therefore, for any function $f \in \mc V$ and any $i=1 \dots n$, by definition of the variational derivative and of the adjoint action, we obtain
\begin{equation}
\begin{split}
    \frac{\delta}{\delta u}(fq_i)&=D_f^*(q_i)+D_{q_i}^*(f) \\
                                            &=D_f^*(q_i) + D_{q_i}(f) \\
                                            &= D_f^*(q_i)+X_f(q_i).
\end{split}
\end{equation}
We know that $L$ is recursion for $p$, meaning that
\begin{equation}
X_p(L)=[D_p,L].
\end{equation}
By Lemma $4.3$, the non-local part of equation $(5.5)$ is:
\begin{equation}
X_p(\sum_{i=1}^n{p_i \partial^{-1} q_i})=\sum_{i=1}^n{ D_p(p_i)\partial^{-1}q_i -p_i\partial^{-1} D_p^*(q_i)}.
\end{equation}
Since $X_p$ is a derivation, $(5.6)$ rewrites into
\begin{equation}
\sum_{i=1}^n{(D_p(p_i)-X_p(p_i))\partial^{-1} q_i}=\sum_{i=1}^n{p_i \partial^{-1} (X_p(q_i)+D_p^*(q_i))}.
\end{equation}
Remembering $(1.11)$, $(1.20)$ and $(5.4)$, we deduce from $(5.7)$ that
\begin{equation}
\sum_{i=1}^n{ \{p_i,p\} \partial^{-1} q_i }=\sum_{i=1}^n{p_i \partial^{-1} \frac{\delta }{\delta u}(pq_i)}.
\end{equation}
Since the variational derivative preserves the parity and since $pq_i$ is odd, we have for all $i=1,...,n$
\begin{equation} 
 \frac{\delta }{\delta u}(pq_i) \in \mc V_{\bar{1}}.
\end{equation}
By Lemma $4.4$, we deduce from equations $(5.8)$ and $(5.9)$ that for all $i=1 \dots n$
\begin{equation}
\{p,p_i\}=\frac{\delta }{\delta u}(pq_i)=0.
\end{equation}
Therefore statement $(2)$ holds, and so does $(3)$ by Lemma $1.10$. Finally, note that $(1)$ and $(3)$ are equivalent since $B$ is a right least common multiple of the operators $\frac{1}{q_i} \partial$.\\
\indent
Let $F \in \mc V$ be such that $p=B(F)$. By Lemma $2.15$ and hereditariness of $L$, it follows that $L$ is recursion for $A(F)$. We are left to check that $A(F)$ is an odd function. From the first line of $(4.39)$, we see that $M_i(F)$ is even for all $i=1,\dots,n$. Therefore, using the second line of $(4.39)$, we get that $A(F) \in \mc V_{\bar{1}}$.
\end{proof}

\begin{theorem}
Let $L \in (\mc V [\partial])_{\bar{0}}+\mc V_{\bar{1}} \partial^{-1} \mc V_{\bar{0}}$ be an integrable rational operator. If $AB^{-1}$ is a right minimal fractional decomposition of $L$ and $F_0 \in Ker B$, then there exists a sequence $F_n \in \mc V, n \geq 0$ such that 
\begin{enumerate}
\item[(1)] $B(F_{n+1})=A(F_n)$ for all $n\geq 0$. \\
\item[(2)] $\{B(F_n),B(F_m)\}=0$ for all $n,m \geq 0$.
\end{enumerate}
Let $m=max \{d(e_k),d(p_i),d(q_i)\}$, where $L=\sum_{k \geq 0}{e_k \partial^k}+\sum_{i=1}^l{p_i\partial^{-1}q_i}$ and where both the $p_i$'s and the $q_i$'s are linearly independent over $\mc C$. Then, if for some $N \geq 0$, $d(B(F_N))>m$, the sequence of differential orders $d(B(F_l))$ goes to $+ \infty$.
\end{theorem}
\begin{proof}
We begin by constructing the sequence $(F_n)$ using Lemma $5.2$. Recall that, if 
\begin{equation}
L=E(\partial)+\sum_{i=1}^l{p_i \partial^{-1} q_i},
\end{equation}
where $l=d(B)$, then $A(Ker B)$ is spanned by the functions $p_i$. Since we assume that $L \in (\mc V [\partial])_{\bar{0}}+\mc V_{\bar{1}} \partial^{-1} \mc V_{\bar{0}}$, we get that $A(Ker B) \subset \mc V_{\bar{1}}$. $L$ is hereditary and recursion for $0=B(F_0)$, hence it is recursion for $A(F_0)$ by Lemma $2.15$. One can apply Lemma $5.2$ with $L$ and $p=A(F_0)$, to find a function $F_1$ such that $A(F_0)=B(F_1)$ and $A(F_1)$ is odd. Iterating the argument, we construct $F_n, n \geq 0$ such that $B(F_{n+1})=A(F_n)$ for all $n \geq 0$. The fact that the functions $B(F_n)$ pairwise commute follows from Corollary $3.18$.
\\
\indent
We now prove the second part of the Theorem, more precisely that if for some $N \geq 0$, $d(B(F_N)) >m$, then for all $n \geq N$, $d(B(F_{n+1}))=d(B(F_n))+d(L)$. This follows by induction on $n$ from $(4.39)$. Indeed, let us assume that for some $n \geq N$, $d(B(F_n)) >m$. Then, from the first line of $(4.39)$, we deduce that $d(M_i(F_n))=d(B(F_n))-1$. Hence, from the second line of $(4.39)$, we deduce that $d(A(F_n))=d(E(B(F_n))=d(B(F_n))+d(L)$, which is what we wanted to show, since $A(F_n)=B(F_{n+1})$ by construction.
\end{proof}

\begin{theorem}
Let $L=E(\partial)+\sum_{i=1}^n{p_i \partial^{-1} q_i}$ be a hereditary rational operator where  $E$ is a even differential operator, $p_i$'s are linearly independent odd functions and $q_i$'s are linearly independent even variational derivatives. Then $L^k$ is weakly non-local and integrable for all $k \geq 1$. Moreover if 
\begin{equation}
L^k=E_k(\partial)+\sum_{i=1}^{n_k}{p_{ki}\partial^{-1} q_{ki}}.
\end{equation}
where $\{p_{k1},\dots,p_{kn_k}\}$ and $\{q_{k1},\dots,q_{kn_k}\}$ are two sets of linearly independent functions and $E_k$ is a differential operator,
then the functions $p_{ki}$ are odd, $E_k$ is even and $q_{ki}$ are even variational derivatives. Finally, for all $k,l \geq 0$ and  for all $i,j \in [1,n_k] \times [1,n_l]$,
\begin{enumerate}
 \item[(a)]$\{p_{ki},p_{lj}\}=0$. \\
\item[(b)] $p_{ki}. q_{lj} \in \partial \mc V$. \\
\item[(c)] $\rho_{lj}$ is a conserved density of $u_t=p_{ki}$ where $q_{lj}=\frac{\delta \rho_{lj}}{\delta u}$.
\end{enumerate}
\end{theorem}
\begin{proof}
Let us define $A$, $B$ and $M_i$ for $i=1 \dots n$ as in $(4.39)$. 
By Proposition $4.11$, $L$ is integrable. Moreover, $L$ is recursion for all the $p_i$'s. Indeed it is obviously recursion for $0$ and $p_i \in A(Ker B)$. Hence by $(2.37)$ $\mc L_{p_i}(L)=0$. Therefore one can iterate Lemma $5.2$ starting from any of the odd functions $p_i$'s. More specifically, for all $i=1 \dots n$ there exists a sequence of functions $F_m^{i} \in \mc V$ such that $p_i=A(F_0^{i})$ and $B(F_{m+1}^{i})=A(F_m^{i})$ for all $m \geq 0$.  Note that $A(F_m^{i}) \in \mc V_{\bar{1}}$ for all $m \geq 0$ and $i=1 \dots n$. For all integer $k \geq 0$ we define  the subspace $\mc W_k \subset \mc V_{\bar{1}}$ to be the span of the functions $A(F_m^{i})$ for $m \leq k$ and $i=1,...,n$. The spaces $\mc W_k$ do not depend on the choices of the sequence $F_m^{i}$. Indeed, given $F_m^{i}$, $F_{m+1}^{i}$ is uniquely defined up to an element of $Ker B$. But we have $A(Ker B)=\mc W_0$. Note that $\mc W_k \subset \mc V_{\bar{1}}$ for all $k \geq 0$. \\
\indent
 Let us prove by induction on $k \geq 1$ the following statements :
First, $ L^k$ is weakly non-local. Secondly, if $\{p_{k1},\dots,p_{kn_k}\}$ and $\{q_{k1},\dots,q_{kn_k}\}$ are two sets of linearly independent functions and $E_k$ is a differential operator such that 
\begin{equation}
L^k=E_k+\sum_{i=1}^{n_k}{p_{ki}\partial^{-1} q_{ki}} \hspace{2 mm},
\end{equation}
then the functions $p_{ki}$ lie in $\mc W_{k-1}$ for all $i$, $E_k$ is even and $q_{ki}$ are even variational derivatives.
The statements hold for $k=1$. Let us assume that they do for $k\geq 1$. By definition of $\mc W_{k-1}$, the functions $p_{ki}$ lie in the image of $B$, and in particular for all $i=1,\dots,n_k$ and $j=1,\dots,n$ there exists a function $v_{ij} \in \mc V$ such that $q_j p_{ki} =(v_{ij})' \in \partial \mc V$. Hence,
\begin{equation}
\sum_{i=1}^{n_k}{\sum_{j=1}^m{p_j\partial^{-1}q_jp_{ki}\partial^{-1} q_{ki}}} \in W.
\end{equation}
By $(5.12)$ and $(5.13)$, $L^{k+1}$ is weakly non-local.  Let us compute the non-local part of $L^{k+1}$.
\begin{equation}
\begin{split}
L^{k+1}& \equiv (E + \sum_{j=1}^n{ p_j \partial^{-1} q_j}) (E_k+\sum_{i=1}^{n_k}{p_{ki}\partial^{-1} q_{ki}} ) \\
            & \equiv \sum_{i=1}^{n_k}{E(p_{ki})\partial^{-1} q_{ki}} + \sum_{i,j} p_j \partial^{-1} (v_{ij})' \partial^{-1}q_{ki} +\sum_{j=1}^n p_j \partial^{-1} E_k^*(q_j) \\
          & \equiv \sum_{i=1}^{n_k}{(E(p_{ki})+ \sum_{j=1}^n{p_j v_{ij}})\partial^{-1} q_{ki}} +
\sum_{j=1}^n{ p_j \partial^{-1} (E_k^*(q_j)-\sum_{i=1}^{n_k}{v_{ij} q_{ki}})}.
\end{split}
\end{equation}
Let $F \in \mc V$ be such that $p_{ki}=B(F)$ (recall that $\mc W_k \subset Im B$). Then by the first line of equation $(4.39)$ and by construction of $v_{ij}$ we have for all $(i,j)\in [1,n_k] \times [1,n]$
\begin{equation}
v_{ij}'=(M_j(F))'.
\end{equation}
 Therefore, by the second line of $(4.39)$, for all $i=1 \dots n_k$,
\begin{equation}
E(p_{ki})+ \sum_{j}{p_j v_{ij}}=A(F) \hspace{2 mm} \text{mod}  \hspace{1 mm} \mc W_{0} .
\end{equation}
Hence 
\begin{equation}
 (L^{k+1})_{\text{nlc}} \in \mc W_{k} \partial^{-1} \mc V.
\end{equation}
$L$ is an even pseudodifferential operator, and so is $L^m$ for any $m\geq 0$. In particular, the local part of $L^{k+1}$, $E_{k+1}$, is an even differential operator. For the same reason, and since $\mc W_{k} \subset \mc V_{\bar{1}}$, we get that 
\begin{equation}
 (L^{k+1})_{\text{nlc}} \in \mc W_{k} \partial^{-1} \mc V_{\bar{0}}.
\end{equation}
$L^{k}$ is integrable, therefore by Proposition $4.11$ and $(5.18)$, the functions $q_{ki}$ are even variational derivatives. 
\\
Let  $k \geq 0$ and $i \in [1,n_k]$. We know that $p_{ki} \in \mc W_{k-1}$. This implies by definition of $\mc W_{k-1}$ that 
$L$ is recursion for $p_{ki}$. Therefore $L^m$ is recursion for $p_{ki}$ for all $m\geq 0$. A direct application of Lemma $5.2$ then gives that $\{p_{ki},p_{mj} \}=0$ and $p_{ki}q_{mj} \in \partial \mc V$ for all $m\geq 0$ and $j \in [1,n_m]$. We are left to prove $(c)$. Let $l\geq 0$, $j \in [1,n_l]$, and $\rho_{lj}$ be such that $q_{lj}=\frac{\delta \rho_{lj}}{\delta u}$. Then we have 
\begin{equation}
\int {p_{ki}\frac{\delta \rho_{lj}}{\delta u}}=0.
\end{equation}
Recalling $(1.14)$, this precisely means that $\rho_{lj}$ is a conserved density of the equation $u_t=p_{ki}$. 
\end{proof}

\begin{remark}
The same statements hold if we switch the parity of the $p_i$'s and the $q_i$'s, but not the parity of $E$.
\end{remark}

\begin{remark}
The ratio of two compatible local Poisson structures is integrable (Remark $3.6$). If furthermore, one assumes that both $H$ and $K$ are odd, and that $K$ has a full kernel spanned by even variational derivatives, then $L=H.K^{-1}$ satisfies the hypothesis of Theorems $5.4$ and $5.5$, therefore it generates an integrable hierarchy of equations, under some assumption on differential orders of the coefficients of $H$ and $K$.
\end{remark}

\section{Examples}
\subsection{$\lambda$-homogeneous equations with linear leading term}
${}$\\
In [SW09], Wang and Sanders give a classification of $\lambda$-homogeneous differential polynomials with linear leading term, i.e. equations of the form
\begin{equation}
u_t=F=u^{(n)}+P,
\end{equation}
where $P$ is a polynomial in $u,\dots,u^{(n-1)}$. Let $\lambda$ and $\mu$ be some constants, then $F$ is called $\lambda-homogeneous$ of weight $\mu$ if it admits the one parameter group of scaling symmetries 
\begin{equation}
(x,t,u) \mapsto (a^{-1}x,a^{-\mu}t,a^{\lambda}u), \hspace{3 mm} a \in \mathbb{R}^{+}.
\end{equation}
Every $\lambda$-homogeneous equation of the form $(6.1)$, modulo homogeneous transformations in $u$, is an equation lying in the hierarchy of one of the $13$ equations displayed in [SW09] (p.103-104). Their corresponding recursion operators can be found in [W02]. 
\\
\\
For the following equations,
\\ 
\begin{equation*}
\begin{split}
u_t&=u'''+3uu' \hspace{50 mm}  \text{(Korteweg-de Vries)} \\
u_t&=u^{(5)}+10uu'''+25u'u''+20u^2u'' \hspace{10 mm} \text{(Kaup-Kuperschmidt)} \\
u_t&=u^{(5)}+10uu'''+10u'u''+20u^2u'' \hspace{10 mm} \text{(Sawada-Kotera)} \\
u_t&=u'''+u'^2 \hspace{53 mm} \text{(Modified KdV)} \\
u_t&=u'''+u'^3 \hspace{53 mm} \text{(Potential modified KdV)} \\
u_t&=u^{(5)}+5u''u'''-5u'^2u'''-5u'u''^2+u'^5 \hspace{2 mm} \text{(Potential Kuperschmidt)},
\end{split}
\end{equation*}
one can check directly using [W02] that their recursion operators lie in $(R[\partial])_{\bar{0}}+R_{\bar{1}}\partial^{-1} \frac{\delta R_{\bar{0}}}{\delta u}$, where $R+R_{\bar{0}} \oplus R_{\bar{1}}$ is the decomposition of the space of differential polynomials into even and odd parts which one obtains by declaring $u$ to be even and $\partial$ to be odd. Moreover, each of these $6$ equations is odd. Therefore, to apply Theorems $5.3$ and $5.4$, we are left to check that the weakly non-local operators are hereditary, which is a tedious but straightforward computation. Since the above equations are odd, we can apply Lemma $5.2$ and initiate the Lenard-Magri scheme with their corresponding recursion operators at themselves. Thus, we conclude that they lie in integrable hierarchies. These abelian subalgebras of $R$ are infinite-dimensional since the orders condition in Theorem $5.3$ is met for $N=1$ or $N=2$ in each case.
\\
\indent
As for the next four equations in the list,
\begin{equation*}
\begin{split}
u_t&=u'''+u'^2 \hspace{80 mm}  \text{(Potential KdV)} \\
u_t&=u^{(5)}+10uu'''+\frac{15}{2}u''^2+\frac{20}{3}u'^3 \hspace{15 mm} \text{(Potential Kaup-Kuperschmidt)} \\
u_t&=u^{(5)}+10uu'''+\frac{20}{3}u'^3 \hspace{40 mm} \text{(Potential Sawada-Kotera)} \\
u_t&=u^{(5)}+5u'u'''+5u''^2-5u^2u'''-20uu'u''-5u'^3+5u^4u' \hspace{2 mm} \text{(Kuperschmidt)},
\end{split}
\end{equation*}
one can check directly using [W02] that their recursion operators lie in $(R[\partial])_{\bar{0}}+R_{\bar{0}}\partial^{-1} \frac{\delta R_{\bar{1}}}{\delta u}$, where $R+R_{\bar{0}} \oplus R_{\bar{1}}$ is the decomposition of the space of differential polynomials into even and odd parts which one obtains by declaring $u$ and $\partial$ to be odd. Moreover, each of these $4$ equations is even. Therefore, to apply Theorems $5.3$ and $5.4$, we are left to check that the weakly non-local operators are hereditary. Similarly, these abelian subalgebras of $R$ are infinite-dimensional since the orders condition in Theorem $5.3$ is met for $N=1$ or $N=2$ in each case.
\\
\indent
The Burgers and the Potential Burgers equations,
\begin{equation*}
\begin{split} 
u_t&=u''+uu' \hspace{5 mm} \text{(Burgers)} \\
u_t&=u''+u'^2 \hspace{5 mm} \text{(Potential Burgers)},
\end{split}
\end{equation*}
admit the following recursion operators
\begin{equation}
\begin{split} 
L_B&=\partial(\partial+u)\partial^{-1} \\
L_{PB}&=\partial+u'.
\end{split}
\end{equation}
Both of these operators are integrable and recursion for $u'$. Moreover, it is clear that they can be applied infinitely many times to $u'$. By Corollary $3.18$ the functions $L_B^n(u'),n\geq 0$ (resp.$L_{PB}^n(u'),n\geq 0$) define an integrable system.
\\
\indent
As for the last equation from the list of [SW09], the Calogero-Degasperis-Ibragimov-Shabat equation :
\begin{equation*}
u_t=u'''+3u^2u''+9uu'^2+3u^4u' \hspace{5 mm} \text{(CDIS)},
\end{equation*}
it admits a rational recursion operator, which is not weakly non-local  :
\begin{equation}
L_{CDIS}=\frac{1}{u} \partial (\partial +2u^2)^{-1}(\partial+u^2-\frac{u'}{u})^2(\partial+2u^2){\partial}^{-1}u.
\end{equation}
Moreover, it is not hard to check that $L_{CDIS}^n(u')$ is well-defined for all $n \geq 0$
Indeed, if in the differential algebra extension of $\mc V$, $\tilde{\mc V}=\mc V[w]$, where $\frac{w'}{w}=u^2$, $L_{CDIS}^n$ rewrites as follows for all $n\geq 0$ :
\begin{equation}
L_{CDIS}^n=\frac{1}{u}\partial \frac{1}{w^2} \partial^{-1} uw \partial^{2n} (\partial uw +(uw)')\frac{1}{u^2} \partial^{-1} u.
\end{equation}
Hence, 
\begin{equation}
L_{CDIS}^n(u')=\frac{1}{u}(\frac{1}{w^2} \int {(uw \partial^{2n+1} uw)(1)})',
\end{equation}
which lies in $\tilde{\mc V}$ since the differential operator $uw\partial^{2n+1}uw$ is skewadjoint for all $n \geq 0$ (the constant coefficient of a skewadjoint differential operator $H$ is a total derivative since $\int{H(1)}=\int{H^*(1)}=-\int{H(1)}$). It is clear from $(6.6)$ that $L_{CDIS}^n(u') \in \mc V \subset \tilde{\mc V}$ for all $n \geq 0$. More precisely, $ \int {(uw \partial^{2n+1} uw)(1)} \in w^2 \mc V$. Finally, one checks that $L_{CDIS}$ is an integrable rational operator and conclude that \textit{CDIS} is an integrable equation by Corollary $3.18$.

\subsection{Krichever-Novikov hierarchy}
hey \\
 In [DS08], Demskoi and Sokolov give a degree $4$ weakly non-local recursion operator $L_{KN}$  for the Krichever-Novikov equation, 
\begin{equation*}
\frac{du}{dt}=u'''-\frac{3{u''}^2}{2u'}+\frac{P(u)}{u'} \hspace{2 mm} \text{(Krichever-Novikov)},
\end{equation*}
where $P$ is a polynomial of degree at most $4$. The recursion operator is of the form 
\begin{equation*}
L_{KN}=\partial^4+a_1\partial^3+a_2\partial^2+a_3\partial+a_4 +G_1 \partial^{-1} \frac{\delta \rho_1}{\delta u}+u'\partial^{-1} \frac{\delta \rho_2}{\delta u}.
\end{equation*}
The space of Laurent differential polynomials in $u$, $\mc A=\mathbb{C}[u^{\pm 1},u'^{\pm 1},...]$ admits a decomposition into even and odd parts $\mc A=\mc A_{\bar{0}} \oplus \mc A_{\bar{1}}$ by declaring $u$ to be even and $\partial$ to be odd. From the explicit formulas given in [DS08], it is straighforward to check that $a_i$ has the same parity as $i$ for $i=1,\dots,4$ and that $\rho_i$ are even for $i=1,2$. Hence, the local part of $L_{KN}$ is even and so are the functions $\frac{\delta \rho_i}{\delta u}$, since variational derivatives preserve parity. Moreover, $G_1$ is the equation (KN) itself, which is odd and so is $u'$. Finally, one checks that $L_{KN}$ is hereditary and apply Theorem $5.3$ (after checking the orders consition in Theorem $5.3$) to conclude that $(KN)$ lies in an infinite dimensional abelian subalgebra of $(\mc A,\{.,.\})$.

\end{document}